\documentclass[twocolumn, 10pt, final]{IEEEtran}

\usepackage{soul}
\usepackage{amsmath}
\usepackage{amssymb}
\usepackage{breqn}
\usepackage{bm}
\usepackage{bbm}
\usepackage{graphicx}
\usepackage{float}
\usepackage{amsthm}
\usepackage{mathrsfs}
\usepackage{algorithm}
\usepackage{dblfloatfix}
\usepackage{algorithmic}
\usepackage{array}
\usepackage{eqparbox}

\usepackage{balance}
\usepackage{multirow}
\usepackage[table,xcdraw]{xcolor}
\usepackage{booktabs}
\usepackage{lipsum}
\usepackage{setspace}
\usepackage{times}
\usepackage{float}
\usepackage{caption} 
\theoremstyle{plain}
\newtheorem{thm}{Theorem}
\newtheorem{lemma}{Lemma}
\newtheorem{corollary}{Corollary}

\usepackage[noadjust]{cite}

\begin{document}

\title{Multitaper Analysis of Evolutionary Spectra from Multivariate Spiking Observations}

\author{Anuththara Rupasinghe and Behtash Babadi \thanks{The authors are with the Department of Electrical \& Computer Engineering, University of Maryland, College Park, MD 20742. E-mails: \{raar, behtash\}@umd.edu. This work has been supported in part by the National Science Foundation Award No. 1807216 and the National Institutes of Health award No. 1U19NS107464-01. This work was presented in part at the IEEE Data Science Workshop, Minneapolis, MN, USA, June 2--5, 2019 \cite{Rupasinghe}.} 
\vspace{-3mm}}

\maketitle

\begin{abstract}
Extracting the spectral representations of the neural processes that underlie spiking activity is key to understanding how the brain rhythms mediate cognitive functions. While spectral estimation of continuous time-series is well studied, inferring the spectral representation of latent non-stationary processes based on spiking observations is a challenging problem. In this paper, we address this issue by developing a multitaper spectral estimation methodology that can be directly applied to multivariate spiking observations in order to extract the evolutionary spectral density of the latent non-stationary processes that drive spiking activity, based on point process theory. We establish theoretical bounds on the bias-variance trade-off of the proposed estimator. Finally, we compare the performance of our proposed technique with existing methods using simulation studies and application to real data, which reveal significant gains in terms of the bias-variance trade-off.
\end{abstract}


\IEEEpeerreviewmaketitle

\section{Introduction}
\IEEEPARstart    
Neural oscillations are known to play a significant role in mediating the cognitive and motor functions of the brain \cite{Schnitzler2005NormalAP, Uhlhaas_Singer, WARD2003553}. The advent of high-density electrophysiology recordings \cite{Vertes_Stackman,Jun_2017, Du_2011} from multiple locations in the brain has opened a unique window of opportunity to probe these oscillations at the neuronal scale. In order to exploit such experimental data for inferring the mechanisms of brain function, spectral analysis techniques tailored for such neuronal spiking data are required \cite{Brown_Kass_2004}.

Existing techniques for spectral analysis of neuronal data use point process theory \cite{Brown_Kass_2004, Truccolo_2005, Paninski_2004} to capture the spiking statistics. Due to the time-domain smoothing procedures used by existing techniques \cite{Smith,lewis2013local,Halliday1999} for recovering the latent processes that drive spiking activity, the resulting power spectral density (PSD) estimates undergo distortion in the spectral domain. Alternative methods to directly estimate the PSD from spiking data have recently been proposed in \cite{Miran,DasPreprint}. 

These existing methods consider univariate spiking observations and assume the latent process to be second-order stationary during the observation period. However, it is known that the brain oscillations that drive neuronal spiking are non-stationary and may exhibit rapid changes corresponding to the brain state or behavioral dynamics \cite{lewis2013local, Grün_2002}. 

Non-stationary time series analysis has been well studied for multivariate continuous signals and various methods have been proposed to quantify the energy-frequency-time distributions {\cite{Huang, Cohen_1995, Priestley_1981, Matz_2006 , Martin_1985, Priestley_1965, Matz_1997, HAMMOND_1996, Das,kim2018state}}. One notable example is the evolutionary power spectral characterization  \cite{Priestley_1965, Matz_1997}, which defines a non-stationary spectral density matrix in order to quantify the local spectral energy distributions at each instant of time for a multivariate process. This characterization has a physical interpretation similar to that of stationary Fourier spectral analysis, and reduces to the classical power spectra when the processes are stationary \cite{Priestley_1965}. A unified approach that considers multivariate spiking observations driven by non-stationary latent processes is lacking, but highly desired due to the emerging demands of modern neuronal data analysis.

In this work, we close this gap by developing a framework to estimate the evolutionary spectral density (ESD) matrix of a multivariate non-stationary latent process, given spiking observations. We model the spiking observations as multiple realizations of point processes with logistic links to the latent continuous processes. We then pose the problem of spectral estimation within a multitapering framework. Multitapering is a widely-used PSD estimation technique with desirable bias-variance trade-off performance \cite{Thomson}, due to the usage of the discrete prolate spheroidal sequences (dpss) \cite{Slepian_1978} as data tapers, which are known for their minimal spectral leakage \cite{percival_walden_1993, Babadi2014ARO, Bronez_1992}.

We employ a state-space model to characterize the dynamics of the evolutionary spectra, with the underlying states pertaining to the eigen-spectra of the multivariate latent processes \cite{Das}. We derive an Expectation-Maximization (EM) algorithm for efficiently computing the maximum a posteriori (MAP) estimate of the latent variables and smoothed states given the spiking observations, which we then use to construct the ESD matrix. We establish theoretical bounds on the bias-variance performance of our proposed estimator, and recover the favorable asymptotic properties of the classical multitaper framework.

We compare the performance of our proposed method to existing techniques through simulation studies and application to experimentally-recorded neuronal data. We present two simulated case studies using non-stationary multivariate autoregressive processes, whose dynamics are inspired by neural oscillations. These studies reveal that the proposed method outperforms two of the widely used methods for deriving spectral representations from spiking data. Finally, we apply our proposed estimator to multi-unit spike and local field potential (LFP) recordings from a human subject undergoing general anesthesia \cite{lewis2013local,Miran}. Our proposed method corroborates existing hypotheses on the relation between the LFP signals and spiking dynamics, by providing a high resolution characterization of the underlying spectrotemporal couplings.

The rest of the paper is organized as follows: We present our problem formulation in Section \ref{sec:prob}, followed by the proposed estimation framework in Section \ref{sec:estimate}. We provide our theoretical results on the bias-variance performance of the proposed estimator in Section \ref{sec:analysis}. Our simulation studies are presented in Section \ref{sec:sim}, followed by application to experimentally-recorded data in Section \ref{sec:data}. Finally, we close the paper by our concluding remarks in Section \ref{sec:con}.

\section{Problem Formulation}\label{sec:prob}

Let $N(t)$ and $H(t)$ denote the point process representing the number of spikes and spiking history of a neuron in $[0,t)$, respectively, where $t  \in [0,T]$ and $T$ denotes the observation duration. The Conditional Intensity Function (CIF) \cite{Daley} of a point process $N(t)$ is defined as:
\begin{equation}
\lambda(t|H_t) := \lim_{\Delta \rightarrow 0} \frac{P[N(t+\Delta)-N(t) = 1|H_t]}{\Delta}.
\label{eq1}
\end{equation}  
To discretize the continuous process, we consider time bins of length $\Delta$, small enough that the probability of having two or more spikes in an interval of length $\Delta$ is negligible. Thus, the discretized point process can be modeled by a Bernoulli process with success probability $\lambda_k : = \lambda(k\Delta|H_{k}) \Delta$, for $1 \le k \le K$, where $K := T/\Delta$ is an integer (with no loss of generality). We refer to $\lambda_k$ as CIF hereafter for brevity.

In a similar fashion, we consider spiking observations from an ensemble of $J$ neurons, with CIFs $\{\lambda_{k,j}\}_{k = 1}^{K}$, for $j=1,2,\cdots,J$. Suppose that for each neuron, $L$ independent realizations of the spiking activity are observed. The collection of the binary spiking observations are represented as $\{n_{k,j}^{(l)}\}_{l, k, j = 1}^{L, K, J}$. We model the $j^{\sf th}$ CIF by a logistic link to a latent random process, $\mathbf{x}_j = [x_{1,j}, x_{2,j},\cdots,x_{K,j}]^{\top}$, which needs not be stationary in general.  Thus, the matrix $\mathbf{X} = [\mathbf{x}_{1}, \; \mathbf{x}_{2}, \cdots, \mathbf{x}_{J}]_{K \times J}$ represents a $J$-variate random process. 

Accordingly, for $1 \leq j \leq J, \; 1 \leq k \leq K$ and $1 \leq l \leq L$, we have 
\begin{equation}\label{eq:logistic}
n_{k,j}^{(l)} \sim \operatorname{Bernoulli}(\lambda_{k,j}),
\end{equation}
where $\lambda_{k,j} = \operatorname{logistic} (x_{k,j}) = 1 / (1 + \exp{(-x_{k,j})})$. Our goal is to estimate the time variant power spectral density of each process $\mathbf{x}_j$, for $1 \leq j \leq J$, and the time variant cross spectra between each pair of processes.  Following the formation of Priestley's Evolutionary Spectra \cite{Priestley_1965}, each random process $x_{k,j}$, with mean $\mu_{k,j}$, will have a representation of the form,
\begin{equation}
x_{k,j} - \mu_{k,j} =  \int_{-\pi}^{\pi} e^{ik\omega} A_{k,j}(\omega) \: dZ_{j}(\omega),
\label{eq2}    
\end{equation}

\noindent where $A_{k,j}(\omega)$ is the time-varying amplitude function and $dZ_{j}(\omega)$ is an orthogonal increment process. To define a discrete-parameter harmonic process, we approximate $Z_{j}(\omega)$ by a jump process over $N$ frequency bins \cite{Miran}, and thereby replace $dZ_{j}(\omega_n)$ with $\frac{\pi}{N} (a_{j,n} + i b_{j,n})$, at $\omega_{n} = n \pi / N, 1 \leq n \leq N-1$. Given the random process is real, using the symmetry $Z_{j}(\omega) = Z_{j}(-\omega)$, we express the discretization of Eq. (\ref{eq2}) as

\begin{equation}
\nonumber    x_{k,j} = \mu_{k,j} + \frac{2\pi}{N} \sum_{n = 1}^{N - 1} (c_{k,j,n} \cos(\omega_n k) - d_{k,j,n} \sin(\omega_n k)),
\label{eq3}
\end{equation}
where $c_{k,j,n} =  A_{k,j}(\omega_n) \:a_{j,n}$ and $d_{k,j,n} =  A_{k,j}(\omega_n) \:b_{j,n}$ are real-valued random variables. Further, the evolutionary spectrum \cite{Priestley_1965} at time $k$ is defined as 
\begin{equation}
\nonumber    \psi_{k,j}(\omega_n) \, d\omega_n = |A_{k,j}(\omega_n)|^2 \: \mathbb{E} |dZ_{j}(\omega_n)|^2.
\label{eq4}
\end{equation}

\noindent Hence, according to our model, the ESD function can be approximated by,
\begin{equation}
\nonumber \psi_{k,j}(\omega_n) = \frac{\pi}{N} \: \mathbb{E} [(c_{k,j,n} + i d_{k,j,n})(c_{k,j,n} - i d_{k,j,n})].
\label{eq5}
\end{equation}

The Spectral Density Matrix of a multivariate stationary random process is defined as, $\bm{\Psi}(\omega) = \frac{1}{2\pi} \sum_{\ell = -\infty}^{\infty} e^{-i\ell\omega} \, \bm{\Gamma}(\ell)$, where $\bm{\Gamma}(\cdot)$ is the covariance matrix of the process \cite{Brockwell}. Further, considering a vector valued orthogonal increment process $\mathbf{Z}(\omega) = [Z_{1}(\omega), Z_{2}(\omega), \cdots , Z_{J}(\omega)]^{\top}$, the spectral density matrix is characterized by $\bm{\Psi} (\omega) d\omega = \mathbb{E}[d\mathbf{Z}(\omega) d\mathbf{Z}(\omega)^{H}]$. We extend this to the evolutionary spectra, and formulate the ESD matrix at time $k$ and frequency $\omega_n$ for our model as,  
\begin{equation}
\nonumber \bm{\Psi}_k(\omega_n) = \frac{\pi}{N} \: \mathbb{E} [(\mathbf{c}_{k,n} + i \mathbf{d}_{k,n})(\mathbf{c}_{k,n} - i \mathbf{d}_{k,n})^\top], 
\label{eq6}
\end{equation}
where $\mathbf{c}_{k,n} = [c_{k,1,n} , \, c_{k,2,n},  \cdots , c_{k,J,n}]^{\top}$ and $\mathbf{d}_{k,n} = [d_{k,1,n} , \, d_{k,2,n}, \cdots , d_{k,J,n}]^{\top}$.

Further, we assume the processes to be quasi-stationary \cite{Das}, the $J$-variate random process to be jointly stationary in windows of length $W$. The total data duration, $K$ is divided into $M$ non overlapping segments of length $W$, with $K = MW$. Thus, the vector process $[x_{k,1}, \, x_{k,2}, \cdots , x_{k,J}]$ is assumed to be jointly stationary for $(m-1)W + 1 \leq k \leq mW$, $ 1 \leq m \leq M$. Simplifying the previous model under this quasi-stationarity assumption, for $(m-1)W + 1 \leq k \leq mW$, $ 1 \leq m \leq M$ we get,  
\begin{equation}
\nonumber  x_{k,j} = \mu_{m,j} + \frac{2\pi}{N} \sum_{n = 1}^{N - 1} (p_{m,j,n} \cos(\omega_n k) - q_{m,j,n} \sin(\omega_n k)),
\label{eq7}
\end{equation}
where $p_{m,j,n}$ and $q_{m,j,n}$ are real-valued random variables.

Defining $\mathbf{X}_{m,j} = [x_{(m-1)W + 1,j}, x_{(m-1)W + 2,j},\cdots, $ $ x_{mW,j}]^{\top}$, $\mathbf{v}_{m,j} = [\frac{N}{2\pi} \mu_{m,j} ,\; p_{m,j,1}, \; q_{m,j,1}, \cdots, p_{m,j,N-1},$ $ q_{m,j,N-1}]^{\top}$, $\Bar{\mathbf{X}}_m = [\mathbf{X}_{m,1}, \mathbf{X}_{m,2}, \cdots ,\mathbf{X}_{m,J}]_{W \times J}$, $\mathbf{V}_m = [\mathbf{v}_{m,1},  \mathbf{v}_{m,2}, \cdots, \mathbf{v}_{m,J}]_{(2N-1) \times J}$ and $\mathbf{A}_m$ as in Eq. (\ref{eq8}), we get  $\mathbf{X}_{m,j} = \mathbf{A}_m \mathbf{v}_{m,j}$ and $\Bar{\mathbf{X}}_m = \mathbf{A}_m \mathbf{V}_m$. Further, we define $\mathbf{p}_{m,n} = [p_{m,1,n} , p_{m,2,n} , \cdots , p_{m,J,n}]^{\top}$, $\mathbf{q}_{m,n} = [q_{m,1,n} , q_{m,2,n} , \cdots , q_{m,J,n}]^{\top}$, $\mathbf{w}_{m,0} =[\mu_{m,1} , \mu_{m,2} , \cdots , \mu_{m,J}]^{\top}$, $\mathbf{w}_{m,n} = [\mathbf{p}_{m,n}^{\top} \,, \, \mathbf{q}_{m,n}^{\top}]^{\top}_{2J \times 1}$ for $1 \leq n \leq N-1$ and $\mathbf{w}_{m} = [\mathbf{w}_{m,0}^{\top} , \mathbf{w}_{m,1}^{\top} , \cdots , \mathbf{w}_{m,N-1}^{\top}]^{\top}_{J(2N-1) \times 1}$. Note that $\mathbf{w}_{m}$ is the vectorization of the matrix $\mathbf{V}_m$ and both are equivalent representations, for the discrete parameter harmonic process that drives the spiking observations in the time window $m$.
\begin{figure*}[!ht]
\begin{align}
\mathbf{A}_m  := \frac{2\pi}{N} 
\resizebox{0.85\textwidth}{!}{$\begin{bmatrix}
1 & \cos\left(\frac{(m-1)W+1}{N} \pi \right) & -\sin\left(\frac{(m-1)W+1}{N}\pi\right) & \dots & \cos\left(\frac{(N-1)((m-1)W+1)}{N} \pi\right) &-\sin\left(\frac{(N-1)((m-1)W+1)}{N} \pi\right) \\
1 & \cos\left(\frac{(m-1)W+2}{N}\pi \right) & -\sin\left(\frac{(m-1)W+2}{N} \pi\right) & \dots & \cos\left(\frac{(N-1)((m-1)W+2)}{N}\pi \right) &-\sin\left(\frac{(N-1)((m-1)W+2)}{N}\pi\right) \\
\vdots & \vdots & \vdots & \ddots & \vdots & \vdots\\
1 & \cos\left(\frac{mW}{N}\pi \right) & -\sin\left(\frac{mW}{N}\pi \right) & \dots & \cos\left(\frac{(N-1)mW}{N} \pi \right) &-\sin\left(\frac{(N-1)mW}{N}\pi \right) \\
\end{bmatrix}$}.
\label{eq8}
\end{align}
\vspace{-4mm}
\end{figure*} 

Accordingly, the ESD matrix,    
\begin{equation}
\bm{\Psi}_m(\omega_n) = \frac{\pi}{N} \: \mathbb{E} [(\mathbf{p}_{m,n} + i \mathbf{q}_{m,n})(\mathbf{p}_{m,n} - i \mathbf{q}_{m,n})^\top],
\label{eq9}
\end{equation}
can be determined if $\mathbb{E}[\mathbf{w}_{m,n} \mathbf{w}_{m,n}^{\top}]$ is estimated for $1 \leq n \leq N-1$. Thus, the task of determining the evolutionary power spectra of the $J$-variate random process can be reduced to computing $\mathbb{E}[\mathbf{w}_{m} \mathbf{w}_{m}^\top]$, for $m=1,2,\cdots,M$, given the ensemble of spiking data $\mathcal{D} =  \{n_{k,j}^{(l)}\}_{l, k, j = 1}^{L, K, J}$. 

\section{Multitaper Estimate of the Spectral Density Matrix}\label{sec:estimate}

\subsection{The Multitaper Framework}\label{proposed_estimator}

Multitapering is a technique widely used in power spectral estimation of stationary random processes, to overcome bias and variance limitations of conventional Fourier analysis \cite{Thomson}. The multitaper spectral estimate of a stationary process, $x_1,x_2, \cdots, x_{K}$ is defined as,

\vspace{-3mm}
\begin{equation}
\begin{split}
S^{\sf mt}(\omega) \quad 
&  = \quad \frac{1}{P} \sum_{p = 1}^{P} \, |y^{(p)}(\omega)|^2 \quad \text{with}, \\
y^{(p)}(\omega) \quad 
& =  \quad  \sum_{k = 1}^{K} \nu_k^{(p)} x_k \, e^{-i \omega k},
\end{split}
\label{eq18}
\end{equation}
\vspace{-2mm}

\noindent where $\nu_k^{(p)}$ is the $k^{\sf th}$ time sample of the $p^{\sf th}$ discrete prolate spheroidal sequence (dpss) \cite{Slepian_1978}, for $ 1 \leq p \leq P$. The dpss tapers are a set of orthogonal tapers that maximally concentrate their spectral power within a design bandwidth of $[-\frac{\xi f_s}{K},\frac{\xi f_s}{K}]$, for some positive constant $\xi$. While originally designed for univariate time series, multitapering has been likewise extended to multivariate time series \cite{Walden}. Given a second order jointly stationary $J$-dimensional random process $\mathbf{x}_1, \mathbf{x}_2, \cdots , \mathbf{x}_{K}$, where, $\mathbf{x}_k = [x_{k,1}, x_{k,2}, \cdots, x_{k,J}]$, the multitaper cross-spectral estimate between the $r^{\sf th}$ process and the $t^{\sf th}$ process $(r,t \in \{1, 2, \cdots, J\})$ has been defined as,  
\vspace{-2mm}
\begin{equation}
\begin{split}
S^{\sf mt}_{r,t}(\omega) \quad 
& = \quad \frac{1}{P} \sum_{p = 1}^{P} \, y_{r}^{(p)}(\omega)( y_{t}^{(p)}(\omega))^* \quad \text{where,} \\
y_{r}^{(p)}(\omega) \quad  
& = \quad  \sum_{k = 1}^{K} \nu_k^{(p)} x_{k,r} \, e^{-i \omega k}. \quad \quad  \quad \quad \quad 
\end{split}
\label{eq19}
\end{equation}
\vspace{-2mm}

The spectral estimate $y_{r}^{(p)}(\omega)( y_{t}^{(p)}(\omega))^*$ corresponding to the $p^{\sf th}$ taper is referred to as the $p^{\sf th}$ eigen-spectral estimate. Given that our observed data are binary random variables parameterized by a function of $\mathbf{X}$ (and not the actual process $\mathbf{X}$), determining the multitaper spectral estimate is not straightforward. However, note that the data log-likelihood,

\vspace{-2mm}
\begin{align}
 \nonumber \log \, f(\mathcal{D} | \mathbf{X}) & = \sum_{j = 1}^{J} \sum_{k = 1}^{K} \sum_{l = 1}^{L} \Big\{n_{k,j}^{(l)} x_{k,j}
- \log \, (1 + \exp(x_{k,j})\Big\} \\
& = \sum_{j = 1}^{J} \sum_{k = 1}^{K} L \Big\{\overline{n}_{k,j} x_{k,j}
\!-\! \log (1 \!+\! \exp(x_{k,j})\Big\}
\label{eq54}
\end{align}
\vspace{-2mm}

\noindent depends on the observations only through the ensemble mean, $\overline{n}_{k,j} := \frac{1}{L} \sum_{l = 1}^{L} n_{k,j}^{(l)}$, resulting $\overline{n}_{k,j}$ to be a sufficient statistic.  Thus, if one could characterize the effect of tapering the time series on the ensemble mean, it would be possible to form the multitaper spectral estimate of the unobserved time series. 

Given that $n_{k,j}^{(l)} \sim \operatorname{Bernoulli}(\lambda_{k,j})$, with $\lambda_{k,j} = \operatorname{logistic}(x_{k,j}) := 1 / (1 + \exp(-x_{k,j})) $, we have $ x_{k,j} = \operatorname{logit}(\lambda_{k,j}) := \log(\lambda_{k,j} / (1 - \lambda_{k,j}))$ and, 
\begin{equation}
\begin{split}
\nonumber \overline{n}_{k,j} \, =\, \frac{1}{L}\sum_{l = 1}^{L} n_{k,j}^{(l)} \quad \rightarrow{} \quad \operatorname{logistic}(x_{k,j}) \quad \quad a.s.
\end{split}
\label{eq20}
\end{equation}
by the strong law of large numbers, for $\overline{n}_{k,j} \neq 0,1$. Further, consider the log-likelihood $\log f( \{n_{k,j}^{(l)}\}_{l = 1}^{L} | x_{k,j} ) = L \Big\{ \overline{n}_{k,j} x_{k,j} - \log \, (1 + \exp(x_{k,j}) \Big\}$, in Eq. (\ref{eq54}). We notice that 
\begin{equation}
\begin{split}
\nonumber \left. {\frac{\partial \log f( \{n_{k,j}^{(l)}\}_{l = 1}^{L} | x_{k,j} )}{\partial  x_{k,j}}}
 \right|_{ x_{k,j} = \operatorname{logit}(\overline{n}_{k,j})} = \quad  0,
\end{split}
\label{eq55}
\end{equation}
which implies that $\operatorname{logit}(\overline{n}_{k,j})$ is the maximum likelihood estimator of $x_{k,j}$. Thus, we take $\operatorname{logit}(\overline{n}_{k,j})$ as an estimator of $x_{k,j}$. By extending this argument to the tapered processes $\{\nu_{k \! \! \! \mod \! W + 1} ^{(p)} x_{k,j} \}_{k, j, p=1}^{K,J,P}$, we can similarly take  
\begin{equation}
\begin{split}
(\overline{n}_{k,j})^{(p)} \quad = \quad \operatorname{logistic} \, (\, \nu_{k \! \! \! \! \! \mod \! W + 1} ^{(p)} \, \operatorname{logit}(\overline{n}_{k,j}) \,),
\end{split}
\label{eq22}
\end{equation}
to be the estimator of the ensemble mean that would have been generated if the random process had been filtered by the $p^{\sf th}$ data taper, for $1 \leq k \leq K$. Note that when $\overline{n}_{k,j} = 0$ or $\overline{n}_{k,j} = 1$, the function $\operatorname{logit}(\overline{n}_{k,j})$ is not defined. In such cases, we directly estimate $(\overline{n}_{k,j})^{(p)}$ by $\overline{n}_{k,j}$.

We thus need to compute the evolutionary spectra corresponding to each of the $P$ tapers, and finally derive the multitaper estimate by averaging the $P$ eigen-spectral estimates. In the next subsection, we consider estimating the power spectral density matrix of the untapered process first, and then extend it to the case of $P$ tapers by replacing the ensemble average of spiking data $\{\overline{n}_{k,j}\}$ with the tapered ensemble mean $\{\overline{n}_{k,j}^{(p)}\}$, for $p=1,2,\cdots,P$. 

\subsection{MAP Estimation of the Parameters via the EM Algorithm}\label{sec:map}

In order to capture the evolution of the spectra, we impose a stochastic continuity constraint on the random variables $\mathbf{w}_{m}$, following the discrete state-space model,
\begin{equation}
\mathbf{w}_{m} = \mathbf{\Phi} \mathbf{w}_{m-1} + \bm{\eta}_m 
\label{eq10}
\end{equation}
where the state transition matrix $\mathbf{\Phi}$ is a constant matrix and $\bm{\eta}_m \sim \mathcal{N}(\mathbf{0}, \mathbf{Q}_m)$. We consider the special case where $\mathbf{\Phi}$ takes the form $\alpha \mathbf{I}$, for simplicity. Nevertheless, $\mathbf{\Phi}$ can also be estimated from data within the same Expectation-Maximization framework that follows next \cite{shumway1982approach}.

The parameters to be estimated are $\bm{\theta} := \{ \mathbf{Q}_m, \, 1 \leq m \leq M \}$, and the observations are binary spiking data $\mathcal{D} = \{ n_{k,j}^{(l)}\}_{k,j,l \, = 1}^{K, J, L}$. Considering the variable $\mathbf{V} = \{\mathbf{V}_{m}, \, 1 \leq m \leq M\}$ to be the hidden data, we aim at recovering $\boldsymbol{\theta}$ via a MAP estimation problem. First, the log-likelihood of the complete data has the form
\begin{equation}
\begin{split}
& \log \, f(\mathcal{D}, \mathbf{V} | \bm{\theta})   = \log \, f(\mathcal{D} | \mathbf{V} , \bm{\theta})  \: + \: \log \, f_{\mathbf{V}|\bm{\theta}}(\mathbf{V}| \bm{\theta}) \\
& = \sum_{j = 1}^{J} \sum_{m = 1}^{M}  \sum_{w = 1}^{W} \sum_{l = 1}^{L} \Big\{n_{(m-1)W+w,j}^{(l)} (\mathbf{A}_m \mathbf{V}_m)_{w,j}  \\
& - \log \, (1 + \exp(\mathbf{A}_m \mathbf{V}_m)_{w,j}))\Big\} - \frac{1}{2} \sum_{m = 1}^{M} \, \Big\{ \log |\mathbf{Q}_m| \\
& +  (\mathbf{w}_{m} -  \mathbf{\Phi} \mathbf{w}_{m-1})^\top \mathbf{Q}_m^{-1} (\mathbf{w}_{m} -  \mathbf{\Phi} \mathbf{w}_{m-1}) \Big\} + C_1,
\end{split}
\label{eq11}
\end{equation}
where $ \mathbf{w}_0 = \mathbf{0}$ and $C_1$ represents terms that are not functions of $\mathcal{D}, \mathbf{V}$ or $\bm{\theta}$. Next, considering that neuronal spiking activity is typically sparse in time and the number of observed realizations ($L$) is limited, an appropriate prior distribution on the parameters helps in reducing the estimation variance. We consider a diagonal covariance matrix $\mathbf{Q}_m$, whose $i^{\sf th}$ diagonal entry is denoted by $Q_{m,i}$. Further, we assume $\mathbf{Q}_m$ to be independent and identically distributed for $1 \leq m \leq M$, with a distribution of the form,
\begin{equation}
\begin{split}
\nonumber f(\mathbf{Q}_m) \propto \resizebox{0.82\columnwidth}{!}{$\displaystyle \exp \!\left (\!\!- \rho \sum_{j = 1}^{2J} \sum_{n = 1}^{ N - 2} \Big ( \log (Q_{m, J(2n-1)+j}) \!-\! \log (Q_{m, J(2n+1)+j}) \Big )^2 \!\right )$}.
\end{split}	
\label{eq12}
\end{equation}

This prior distribution encourages continuity in log scale of the spectral estimates corresponding to adjacent frequency bins of each CIF, and can be controlled by appropriately selecting the hyper-parameter $\rho$. Accordingly, considering $(\mathcal{D}, \mathbf{V} )$ to be the set of complete data, the joint distribution of the MAP estimation problem is formulated as,
\begin{align}
\nonumber \log &f(\mathcal{D}, \mathbf{V}, \bm{\theta}) =  \log \, f_{\mathbf{V}| \bm{\theta}}(\mathbf{V}| \bm{\theta}) + \log \, f(\bm{\theta}) + C_2\\
\nonumber & \resizebox{0.96\columnwidth}{!}{$\displaystyle = - \frac{1}{2} \sum_{m = 1}^{M} \Big \{ \log |\mathbf{Q}_m| +  (\mathbf{w}_{m} \!-\! \mathbf{\Phi} \mathbf{w}_{m-1})^\top \mathbf{Q}_m^{-1} (\mathbf{w}_{m} \!-\!  \mathbf{\Phi} \mathbf{w}_{m-1}) \Big \}$} \\
\nonumber &  \resizebox{0.94\columnwidth}{!}{$\displaystyle - \rho \sum_{m = 1}^{M} \sum_{j = 1}^{2J} \sum_{n = 1}^{ N - 2} \Big ( \log (Q_{m, J(2n-1)+j}) - \log (Q_{m, J(2n+1)+j}) \Big )^2 + C_3$},
\end{align}
where $C_2$ and $C_3$ represent terms that do not depend on $\bm{\theta}$. We next construct the EM algorithm for solving the MAP estimation problem:
\medskip

\subsubsection{E Step}

Suppose that the current estimate of $\bm{\theta}$ at the $r^{\sf th}$ iteration is given by $\widehat{\bm{\theta}}^{(r)}$.  Then, the Q-function 
\begin{equation}
\mathcal{Q} ^{(r)} := \mathbb{E}[\log \,f(\mathcal{D}, \mathbf{V}, \bm{\theta}) | \mathcal{D}, \widehat{\bm{\theta}}^{(r)}]
\end{equation}
can be evaluated if the conditional expectations 
\begin{align}
\nonumber & \mathbf{w}_{m|M} := \mathbb{E}[\mathbf{w}_{m}| \mathcal{D}, \widehat{\bm{\theta}}^{(r)}]\\
\nonumber & \mathbf{\Sigma}_{m|M} := \mathbb{E}[(\mathbf{w}_{m} \!-\! \mathbf{w}_{m|M})(\mathbf{w}_{m} \!-\! \mathbf{w}_{m|M})^{\top}| \mathcal{D}, \widehat{\bm{\theta}}^{(r)}]\\
\nonumber & \mathbf{\Sigma}_{m, m-1|M} := \mathbb{E}[(\mathbf{w}_{m} \!-\! \mathbf{w}_{m|M})(\mathbf{w}_{m-1} \!-\! \mathbf{w}_{m-1|M})^{\top}| \mathcal{D}, \widehat{\bm{\theta}}^{(r)}]
\end{align}
are known. To compute the conditional expectations $\mathbf{w}_{m|M}$, $\mathbf{\Sigma}_{m|M}$ and $\mathbf{\Sigma}_{m, m-1|M}$ we utilize the Fixed Interval Smoothing \cite{RAUCH_1965} and the Covariance Smoothing \cite{JONG} algorithms. However, considering that the forward model is not Gaussian in our formulation, we cannot directly use these algorithms to estimate $\mathbf{w}_{m|m}$ and $\mathbf{\Sigma}_{m|m}$. 

Hence, we employ an alternative method to estimate these conditional moments, utilizing the distribution $f({\{\mathbf{V}_s\}_{s=1}^m| \mathcal{D}_1^m, \hat{\bm{\theta}}^{(r)}})$, where $\mathcal{D}_1^m = \{ n_{k,j}^{(l)}\}_{k,j,l \, = 1}^{mW, J, L}$. From Bayes's theorem we see that  this is proportional to the product of the two distributions, $f(\mathcal{D}_1^m |\{\mathbf{V}_s\}_{s=1}^m, \hat{\bm{\theta}}^{(r)})$ and $f(\{\mathbf{V}_s\}_{s=1}^m| \hat{\bm{\theta}}^{(r)})$, which are Binomial and Gaussian distributed, respectively. Accordingly, we see that the distribution $\{\mathbf{V}_s\}_{s=1}^m|\mathcal{D}_1^m, \hat{\bm{\theta}}^{(r)}$ is unimodal, and hence we approximate it by a multivariate Gaussian, and derive the mean of the distribution, $\mathbf{w}_{m|m}^{(r)}$ by the mode of $\log \, f({\{\mathbf{V}_s\}_{s=1}^m|\mathcal{D}_1^m, \hat{\bm{\theta}}^{(r)}})$:

\vspace{-2mm}
\begin{align}
\nonumber \mathbf{w}_{m|m}^{(r)} & = \underset{\mathbf{w}_m}{\textrm{argmax}} \; \log \, f({\{\mathbf{V}_s\}_{s=1}^m|\mathcal{D}_1^m, \hat{\bm{\theta}}^{(r)}})  \\
\nonumber & = \underset{\mathbf{w}_m}{\textrm{argmax}} \bigg( \sum_{j = 1}^{J} \sum_{s = 1}^{m}  \sum_{w = 1}^{W}  L \,\Big\{\overline{n}_{(s-1)W+w,j} (\mathbf{A}_s \mathbf{V}_s)_{w,j} \\ 
\nonumber & - \log \, (1 + \exp(\mathbf{A}_s \mathbf{V}_s)_{w,j}))\big\} - \frac{1}{2} \sum_{s = 1}^{m}\, \Big\{ \log |\mathbf{Q}_s^{(r)}| \\ 
&  +  (\mathbf{w}_{s} \!-\!  \mathbf{\Phi} \mathbf{w}_{s-1})^\top (\mathbf{Q}_s^{(r)})^{-1} (\mathbf{w}_{s} \!-\!  \mathbf{\Phi} \mathbf{w}_{s-1}) \Big\} \bigg),
\label{eq14}
\end{align}
\vspace{-3mm}

\noindent and its covariance, $\mathbf{\Sigma}_{m|m}^{(r)}$ by the negative of the inverse Hessian of $\log  f({\{\mathbf{V}_s\}_{s=1}^m|\mathcal{D}_1^m, \hat{\bm{\theta}}^{(r)}})$. Observing that the objective function is a combination of convex functions and is differentiable, we perform the above optimization for $\mathbf{w}_{m|m}^{(r)}$ using the Newton-Raphson method. Further, we concurrently estimate $\mathbf{\Sigma}_{m|m}^{(r)}$ using the Hessian matrix of the objective.

\subsubsection{M Step}
Noticing that the Q-function $\mathcal{Q} ^{(r)}$ is separable in terms of $\mathbf{Q}_m$'s, we can update $\mathbf{Q}_m$ independently, for $1 \leq m \leq M$:
\begin{equation}
\begin{split}
& \mathbf{Q}_m^{(r+1)} = \underset{\mathbf{Q}_m}{\textrm{argmax}} \quad \mathcal{Q} ^{(r)}.
\end{split}
\label{eq15}
\end{equation}
The partial derivative of the log-likelihood function with respect to each diagonal element of $\mathbf{Q}_m$ takes the form,
\begin{equation}
\nonumber \resizebox{\columnwidth}{!}{$\displaystyle \frac{\partial \mathcal{Q} ^{(r)}}{\partial Q_{m,i}} = - \frac{1}{2} \left \{ \frac{1}{Q_{m,i}} \!-\! \frac{P_{m,i}}{Q_{m,i}^2} \right \} \!-\! \frac{2 \rho}{Q_{m,i}} \log \left ( \frac{Q_{m,i}^2}{Q_{m,i-2J} Q_{m,i+2J}} \right ),$}
\label{eq16}
\end{equation}
where $P_{m,i}$ is the $i^{\sf th}$ diagonal element of the matrix
\begin{align}
\nonumber \mathbf{P}_m :=& \, \, \mathbb{E}[ (\mathbf{w}_{m} -  \mathbf{\Phi} \mathbf{w}_{m-1}) (\mathbf{w}_{m} -  \mathbf{\Phi} \mathbf{w}_{m-1})^\top | \mathcal{D}, \widehat{\bm{\theta}}^{(r)}]\\
\nonumber =& \, \, \resizebox{0.885\columnwidth}{!}{$\displaystyle \mathbf{\Sigma}_{m|M} \!+\! \mathbf{w}_{m|M} \mathbf{w}_{m|M}^\top \!+\! \mathbf{\Phi}(\mathbf{\Sigma}_{m-1|M} \!+\!  \mathbf{w}_{m-1|M}  \mathbf{w}_{m-1|M}^\top)\mathbf{\Phi}^\top$}\\
\nonumber &- (\mathbf{\Sigma}_{m,m-1|M}  + \mathbf{w}_{m|M}  \mathbf{w}_{m-1|M}^\top)\mathbf{\Phi}^\top\\
& -\mathbf{\Phi} (\mathbf{\Sigma}^\top_{m,m-1|M}  + \mathbf{w}_{m-1|M}\mathbf{w}_{m|M}^\top ).
\label{eq:P_m}
\end{align}
Accordingly, we employ the multivariate Newton-Raphson algorithm to perform this maximization and derive the updates for $\mathbf{Q}_m$, $1 \leq m \leq M$. Following convergence, we use the final estimates of  $\mathbf{w}_{m|M}$ and $\mathbf{\Sigma}_{m|M}$ derived through the above EM procedure to compute 
\begin{equation}
\mathbf{R}_m := \mathbb{E}[\mathbf{w}_{m}\mathbf{w}_{m}^{\top} | \mathcal{D}, \bm{\theta}] = \mathbf{\Sigma}_{m|M} + \mathbf{w}_{m|M} \mathbf{w}_{m|M}^{\top},
\end{equation}
for $1 \leq m \leq M$. Then, for $1 \leq n \leq N-1$, the ESD matrix according to Eq. (\ref{eq9}) is estimated as
\begin{equation}
\begin{split}
\widehat{\bm{\Psi}}_m(\omega_n) = & \frac{\pi}{N} \Big \{ \left({\mathbf{R}_m^n}_{(1:J,1:J)} + {\mathbf{R}_m^n}_{(J+1:2J,J+1:2J)} \right)\\ 
& + i\left({\mathbf{R}_m^n}_{(J+1:2J,1:J)} - {\mathbf{R}_m^n}_{(1:J,J+1:2J)}\right) \Big\},
\end{split}
\label{eq17}
\end{equation}
where $\mathbf{R}_m^n:= \mathbb{E}[\mathbf{w}_{m,n} \mathbf{w}_{m,n}^{\top}| \mathcal{D}, \bm{\theta}]$ is a submatrix of $\mathbf{R}_m$ defined as:
\begin{align}
\nonumber \mathbf{R}_m^n := ({\mathbf{R}_m})_{(J(2n-1)+1:J(2n+1),J(2n-1)+1:J(2n+1))}.
\end{align}

\setlength{\textfloatsep}{10pt}
{ \renewcommand\baselinestretch{0.8}\selectfont
	\begin{algorithm} [t!]
		\textbf{Inputs:} Ensemble averages of the spiking observations $\{ \overline{n}_{k,j} \}_{k,j = 1}^{K, J}$, hyper-parameters $\rho$ and $\alpha$, maximum number of EM iterations $R_{\max}$.   \\
		\textbf{Outputs:} Estimates of the ESD matrices $\widehat{\bm{\Psi}}_m(\omega_n) $ for $ 1 \leq m \leq M$, $1 \leq j \leq N -1 $\\
		\textbf{Initialization:} Initial choice of $\mathbf{Q}_m^{(0)}$, $\mathbf{w}_{0|0} = \bm{0}$, $\mathbf{\Sigma}_{0|0} = \bm{0}$, $r = 1$.\\
		\vspace{-3mm}
		\begin{algorithmic}[1]
			
			\FOR{$r \leq R_{\max}$}
			
			\STATE Forward filter, for $ m = 1, 2, \dots, M$ \\
			  \quad \quad  $\mathbf{w}_{m|m-1} =  \mathbf{\Phi} \mathbf{w}_{m-1|m-1}$ \\
			  \quad \quad $\mathbf{\Sigma}_{m|m-1} =  \mathbf{\Phi} \mathbf{\Sigma}_{m-1|m-1} \mathbf{\Phi}^\top$ + $\mathbf{Q}_m^{(r)}$ \\
			  \quad \quad Compute $\mathbf{w}_{m|m}$ and $\mathbf{\Sigma}_{m|m}$ using the Newton's method as described in Eq. (\ref{eq14}).\\ 
			 \STATE Backward smoother, for $m = M - 1, M - 2, \dots, 1$ \\
			 \quad \quad $\mathbf{B}_{m} = \mathbf{\Sigma}_{m|m} \mathbf{\Phi}^\top \mathbf{\Sigma}_{m+1|m}^{-1}$\\
			 \quad \quad $\mathbf{w}_{m|M} = \mathbf{w}_{m|m} + \mathbf{B}_{m} (\mathbf{w}_{m+1|M} -  \mathbf{w}_{m+1|m})$ \\
			 \quad \quad $\mathbf{\Sigma}_{m|M} = \mathbf{\Sigma}_{m|m} + \mathbf{B}_{m}(\mathbf{\Sigma}_{m+1|M} - \mathbf{\Sigma}_{m+1|m})\mathbf{B}_{m}^\top$ \\
			 \STATE Covariance smoother, for $m = M-1, M-2, \dots, 1$\\
			 \quad \quad $\mathbf{\Sigma}_{m, m-1|M} = \mathbf{\Sigma}_{m|M}^\top \mathbf{B}_{m-1}^\top$ \\  
			 \STATE Update the $\mathbf{Q}_m$'s independently, for $ m = 1, 2, \dots, M$ using the multivariate Newton-Raphson method to solve \\
			 \quad \quad  $\mathbf{Q}_m^{(r+1)} = \underset{\mathbf{Q}_m}{\textrm{argmax}} \quad \mathcal{Q} ^{(r)}$.
			\STATE Set $r \leftarrow r + 1$ 
			\ENDFOR
			\FOR {$ 1 \leq m \leq M$}
			\STATE $\mathbf{R}_m = \mathbf{\Sigma}_{m|M} + \mathbf{w}_{m|M} \mathbf{w}_{m|M}^\top$ \\
			\FOR{ $1 \leq n \leq N-1 $}
			\STATE $\mathbf{R}_m^n = (\mathbf{R}_m)_{(J(2n-1)+1:J(2n+1),J(2n-1)+1:J(2n+1))}$. \\
			\STATE $\widehat{\bm{\Psi}}_m(\omega_n) = \frac{\pi}{N}\left\{ \mathbf{R}^n_{m{(1:J,1:J)}} + \mathbf{R}^n_{m{(J+1:2J,J+1:2J)}} \right.$\\
			$\quad \quad \quad \quad \left.+ \;i \left (\mathbf{R}^n_{m{(J+1:2J,1:J)}} - \mathbf{R}^n_{m(1:J,J+1:2J)} \right) \right\}$.
			\ENDFOR
			\ENDFOR

			\STATE Return $\widehat{\bm{\Psi}}_m(\omega_n)$ for $1 \leq m \leq M$, $1 \leq n \leq N-1$
			\caption{Estimation of the Evolutionary Spectral Density Matrix via the EM Algorithm}
			\label{Alg:Algorithm1}
		\end{algorithmic}
	\end{algorithm}}

An implementation of this estimation procedure is outlined in Algorithm \ref{Alg:Algorithm1}. As mentioned in Section \ref{proposed_estimator}, the same EM procedure can be carried out for $\{\overline{n}_{k,j}^{(p)}\}$, for $p=1,2,\cdots,P$ in order to estimate the multivariate eigen-spectra. Finally, the multitaper spectral density matrix is formed by averaging the eigen-spectral estimates as outlined in Algorithm \ref{Alg:Algorithm2}. We refer to our proposed algorithm as the Point Process Multitaper Evolutionary Spectral Density (PPMT-ESD) estimator.

\setlength{\textfloatsep}{4pt}
{ \renewcommand\baselinestretch{0.8}\selectfont        
	\begin{algorithm} [t!]
		\textbf{Inputs:} Collection of ensemble averages of the observations $\{ \overline{n}_{k,j} \}_{k,j = 1}^{K, J}$, the set of $P$ dpss tapers of length $W$ $ \{ \nu_w^{(p)} \}_{w,p = 1}^{W, P}$ \\
		\textbf{Outputs:} The multitaper estimates of the ESD matrices $\widehat{\bm{\Psi}}^{\sf mt}_m(\omega_n) $ for $ 1 \leq m \leq M$, $1 \leq n \leq N -1 $\\
		\vspace{-3mm}
		\begin{algorithmic}[1]
			
			\FOR{$ p = 1, 2, \cdots, P$}
			
			\FOR{$1 \leq w \leq W, \, \, 1 \leq m \leq M, \, \, 1 \leq j \leq J$}
			\STATE $ k = ((m-1)W + w)$ \\ 
			\IF{$\overline{n}_{k,j} \neq 0$ and $\overline{n}_{k,j} \neq 1$}
			\STATE $(\overline{n}_{k,j})^{(p)} \quad = \quad \operatorname{logistic} \, (\,\operatorname{logit}(\overline{n}_{k,j}) \, \nu_w^{(p)}\,)$
			\ELSE
			\STATE $(\overline{n}_{k,j})^{(p)} \quad = \quad \overline{n}_{k,j} $
			\ENDIF
			\ENDFOR
			\STATE Compute the $p^{\sf th}$ tapered spectral density matrix estimate, $\widehat{\bm{\Psi}}^{(p)}_m(\omega_n) $ for $ 1 \leq m \leq M$, $1 \leq n \leq N -1 $, using Algorithm \ref{Alg:Algorithm1}, with $\{ \overline{n}_{k,j}^{(p)} \}_{k,j = 1}^{K, J}$ as the input collection of ensemble averages.
			\ENDFOR
			
			\FOR{ $1 \leq m \leq M, \, \, 1 \leq n \leq N-1$}
			\STATE $\widehat{\bm{\Psi}}^{\sf mt}_m(\omega_n) = \frac{1}{P}\sum_{p = 1}^{P} \, \widehat{\bm{\Psi}}^{(p)}_m(\omega_n)$ 
			\ENDFOR

			\STATE return $\widehat{\bm{\Psi}}^{\sf mt}_m(\omega_n)$ for $1 \leq m \leq M$, $1 \leq n \leq N-1$
			\caption{Estimation of the Multitaper Evolutionary Spectral Density Matrix}
			\label{Alg:Algorithm2}
		\end{algorithmic}
	\end{algorithm}}
	
\vspace{-2mm}
\section{Theoretical Analysis}\label{sec:analysis}
\vspace{-2mm}

In this section we derive bounds on the bias and variance of the proposed PPMT-ESD estimator. We first briefly review the corresponding bounds for the direct multitaper estimator (Eq. (\ref{eq18})) of a time-series. As proven in \cite{Rosenblatt}, the bias and variance of the multitaper estimate of a stationary process $x_1, x_2, \cdots, x_{K}$ with a uniformly continuous PSD $S(\omega)$ are bounded as follows:  
\begin{equation}
\begin{split}
|\operatorname{bias}(S^{\sf mt}(\omega))| \, 
&\leq \, (\underset{\omega}{\text{sup}} \, S(\omega)) \left\{ 1 - \frac{1}{P} \sum_{p = 1}^{P}c_p \right\} + o(1), \\
 \operatorname{Var}(S^{\sf mt}(\omega)) 
& = \{ 1 + \beta (\omega)\} \frac{1}{P^2} \sum_{p = 1}^{P}c_p^2 \, S(\omega)^2 \\
& + \mathcal{O}\left( \frac{1}{P} \sum_{p = 1}^{P} (1 - c_p)  \right) + o(1),
\end{split}
\label{eq45}
\end{equation}
as $K \rightarrow \infty$. Here $c_p$ is the eigenvalue corresponding to the taper $ \nu^{(p)}$ and $\beta (\omega) = 0$ if $\frac{\omega}{2\pi} \neq 0, 1/2 \, \operatorname{mod} 1$ and is equal to $1$ otherwise. It is evident that the multitaper estimator $\widehat{S}^{\sf mt}(\omega)$ is asymptotically unbiased.

We state our main theorem for a univariate second-order stationary process $x_1,x_2,\cdots,x_K$, corresponding to the special case of $J=1$, for the clarity of exposition. We later on provide extensions to the multivariate and quasi-stationary cases. In order to apply the same treatment to the case of multitaper estimate of the evolutionary spectra from spiking observations, we need to make two extra technical assumptions. 

\emph{Assumption (1).} From Eq. (\ref{eq54}), the data likelihood for the univariate case can be expressed as 
\begin{equation}
\nonumber f(\mathcal{D}) := \int \resizebox{0.8\columnwidth}{!}{$\displaystyle \exp\left(\sum_{k = 1}^{K} L \Big\{\overline{n}_{k} x_{k}
- \log \, (1 + \exp(x_{k})\Big\}\right) \prod_{k=1}^{K} dx_{k}$},
\end{equation}
where $\overline{n}_k := \frac{1}{L} \sum_{l=1}^L n_k^{(l)}$. Given the nonlinear functional form of the integrand, we consider the saddle point approximation \cite{Goutis_1999} of the integral, and take $\operatorname{logit}(\overline{n}_{k})$ as an estimator of $x_{k}$. Under this approximation, the multitaper spectral estimate is given by
\vspace{-2mm}
\begin{equation}\label{eq23}
\widehat{S}^{\sf mt}(\omega) =\frac{1}{P} \sum_{p = 1}^{P} \, |\widehat{y}^{(p)}(\omega)|^2, \quad \textrm{where,}
\end{equation}
\begin{equation} 
\widehat{y}^{(p)}(\omega) := \sum_{k = 1}^{K} \nu_k^{(p)} \, \operatorname{logit}(\overline{n}_{k}) \, e^{-i \omega k}.
\end{equation}
\vspace{-2mm}

\emph{Assumption (2).} We assume that $|x_k| \; \le B$ for all $k=1,2,\cdots,K$, for some fixed upper bound $B$. In defining the bias and variance, we condition the expectations on the event $A := \{\overline{n}_{k} \; | \; \overline{n}_{k} \neq 0, \, \overline{n}_{k} \neq 1, \, 1 \leq k \leq K \}$, which is highly probable due to the absolute boundedness of $x_{k}$, for large $L$ (See Appendix \ref{Appendix_A} for details). We denote the conditional bias and variance by $\operatorname{bias}_A(\cdot)$ and $\operatorname{Var}_A (\cdot)$, respectively. Note that for the multivariate case, we naturally extend this assumption to $|x_{k,j}| \le B$ for all $k,j$, and define the set $A$ as $A := \{\overline{n}_{k,j} \; | \; \overline{n}_{k,j} \neq 0, \, \overline{n}_{k,j} \neq 1, \, 1 \leq k \leq K, 1 \leq j \leq J \}$.

\begin{thm}[Univariate Case]\label{thm1}
Under the Assumptions (1) and (2) and for sufficiently large $L$, the conditional bias and variance of $\widehat{S}^{\sf mt}(\omega)$ in Eq. (\ref{eq23}) is bounded with respect to those of the direct multitaper estimate $S^{\sf mt}(\omega)$ given in Eq. (\ref{eq45}) as:
\begin{equation}
\begin{split}
|\operatorname{bias}_A(\widehat{S}^{\sf mt}(\omega))| \leq   g_1 K \,\frac{\log L}{\sqrt{L}} \,  + |\operatorname{bias}(S^{\sf mt}(\omega))|,
\end{split}
\label{eq24}
\end{equation}
\begin{equation}
\begin{split}
\operatorname{Var}_A (\widehat{S}^{\sf mt}(\omega)) \leq \, \Bigg \{  \, g_2 K \, \frac{\log L}{\sqrt{L}} \, + \sqrt{ \operatorname{Var}(S^{\sf mt}(\omega))} \Bigg\}^2
\end{split}
\label{eq25}
\end{equation}
where $g_1$ and $g_2$ are bounded constants depending on $B$, $K$ and $L$.
\end{thm}

\begin{proof}[Proof of Theorem \ref{thm1}]
The proof of Theorem \ref{thm1} is given in Appendix \ref{Appendix_A}.
\end{proof}

\noindent{\bf Remark.} In words, Theorem \ref{thm1} states that the cost of the indirect access to the process $\{x_k\}_{k=1}^K$ through spiking observations with $L$ trials appears as excess terms in both the bias and variance, which would go to zero as $\frac{L}{K^2 \log^2 L} \rightarrow \infty$. Hence, for large enough number of realizations $L$, one can expect a performance close to the direct multitaper estimate of $\{x_k\}_{k=1}^K$. The result of Theorem \ref{thm1} can be extended to the case of unconditional expectations, if the inverse mapping $\max(B, \, \min(-B, \, \operatorname{logit} (\overline{n}_{k}) ))$ is adopted instead of $ \operatorname{logit} (\overline{n}_{k}) $. Then, the bounds follow directly even without conditioning on $A$. While Assumption (2) on the boundedness of the time-series is natural in practical scenarios, as long as $B \leq \epsilon \log L$, for some $\epsilon < 1/2$, the excess bias and variance terms will be bounded by $\mathcal{O}(\log^2 L / L ^{1/2 - \epsilon})$, which will converge to zero asymptotically. Thus, the bias and variance of the estimator $\widehat{S}^{\sf mt}(\omega)$ will converge to those of $S^{\sf mt}(\omega)$ even under the milder condition, $|x_k|  \;  \leq \,  \epsilon \log L$ for $ 1 \leq k \leq K$, and some $\epsilon < 1/2$. 

The following corollary extends Theorem \ref{thm1} to the multivariate case:

\begin{corollary}[Stationary Multivariate Case]
\label{corollary1}
Consider a second order jointly stationary J-variate process $\{ \mathbf{x}_k\}_{k=1}^K$, where, $\mathbf{x}_k = [x_{k,1}, x_{k,2}, \cdots, x_{k,J}]^{\top}$. Suppose that the observations are binary spiking data,  $\{ n_{k,j}^{(l)} \}_{k,j,l = 1}^{K, J, L}$ with $ n_{k,j}^{(l)} \sim \operatorname{Bernoulli}({\lambda}_{k,j})$ and $\overline{n}_{k,j} := \frac{1}{L} \sum_{l = 1}^{L} n_{k,j}^{(l)}$, where ${\lambda}_{k,j} = \operatorname{logistic} \, (x_{k,j})$. Then, under Assumptions (1) and (2), the bias and variance of the multitaper cross-spectral estimate between the $r^{\sf th}$ and $t^{\sf th}$ processes, given by 

\vspace{-2mm}
\begin{equation}
\begin{split}
\nonumber \widehat{S}^{\sf mt}_{r,t}(\omega) \quad 
& = \quad \frac{1}{P} \sum_{p = 1}^{P} \, \widehat{y}_{r}^{(p)}(\omega)( \widehat{y}_{t}^{(p)}(\omega))^*, \quad \text{with} \\
\widehat{y}_{\ell}^{(p)}(\omega) \quad  
& = \quad  \sum_{k = 1}^{K} \nu_k^{(p)} \, \operatorname{logit} (\overline{n}_{k,\ell})\, e^{-i \omega k},
\end{split}
\label{eq46}
\end{equation}
\vspace{-2mm}

are bounded as follows:
\begin{equation}
\begin{split}
\nonumber |\operatorname{bias}_A(\widehat{S}^{\sf mt}_{r,t}(\omega))| \quad \leq \quad g'_1 K\frac{\log L}{\sqrt{L}}  + |\operatorname{bias}(S^{\sf mt}_{r,t}(\omega))|,
\end{split}
\label{eq47}
\end{equation}
\begin{equation}
\begin{split}
\nonumber \operatorname{Var}_A (\widehat{S}^{\sf mt}_{r,t}(\omega)) \quad \leq \quad \Bigg \{ g'_2 K \, \frac{\log L}{\sqrt{L}} +   \sqrt{ \operatorname{Var}(S^{\sf mt}_{r,t}(\omega))} \Bigg\}^2,
\end{split}
\label{eq48}
\end{equation}
where $g'_1$ and $g'_2$ are bounded constants and depend only on $B$, $K$, and $L$.
\end{corollary}
\begin{proof}
The proof is given in Appendix \ref{Appendix_B}.
\end{proof}

It is not difficult to verify that $\operatorname{bias}(S^{\sf mt}_{r,t}(\omega))$ and $\operatorname{Var}(S^{\sf mt}_{r,t}(\omega))$ can be upper bounded in a similar fashion to Eq. (\ref{eq45}), with the true cross-spectra $S_{r,t}(\omega)$ replacing $S(\omega)$. Before extending the result of Corollary \ref{corollary1} to the quasi-stationary case, we need an additional assumption:

\emph{Assumption (3).} Given that Corollary \ref{corollary1} holds for large $L$, in this regime we relax the prior distribution on $\mathbf{Q}_m$ to be flat, i.e., $f(\mathbf{Q}_m) \propto 1$. Recall that the rationale for using a prior on $\mathbf{Q}_m$ in Section \ref{sec:map} was to reduce the variance of the estimates in the low spiking regime, i.e., small $L$.

Finally, combining Corollary \ref{corollary1} and the treatment of \cite{Das}, we have the following corollary on the bias and variance of the PPMT-ESD estimator:

\begin{corollary}[Quasi-stationary Multivariate Case]\label{corollary2}
Suppose that the J-variate process in Corollary \ref{corollary1} is quasi-stationary (jointly stationary within consecutive windows of length $W$). Let $\Psi_{m,r,t}(\omega_n)$ be the cross-spectra between the $r^{\sf th}$ and $t^{\sf th}$ processes over window $m$, $1 \leq m \leq K/W$, and $\widehat{\Psi}^{\sf mt}_{m,r,t}(\omega_n)$ be the corresponding multitaper estimate obtained from spiking observations. Then, under Assumptions (1)--(3), the bias and variance of the proposed PPMT-ESD estimator at window $m$ can be bounded as, 
\begin{equation}
\begin{split}
\nonumber & |\operatorname{bias}_A(\widehat{\Psi}^{\sf mt}_{m,r,t}(\omega_n))| \leq  \resizebox{0.6\columnwidth}{!}{$\displaystyle  g''_1(\omega_n) W \, \frac{\log L}{\sqrt{L}} + |\Psi_{m,r,t}(\omega_n)| |1 \!-\! \kappa_m(\omega_n)|$} \\
& + \kappa_m(\omega_n) \left \{\underset{\omega}{\emph{sup}} \left \{ |\Psi_{m,r,t}(\omega)| \right \} \left( 1 - \frac{1}{P} \sum_{p = 1}^{P} c_p \right)  + o(1) \right \},\\
\end{split}
\label{eq49}
\end{equation}
\begin{equation}
\begin{split}
\nonumber & \operatorname{Var}_A (\widehat{\Psi}^{\sf mt}_{m,r,t}(\omega_n)) \! \leq \! \resizebox{0.64\columnwidth}{!}{$\displaystyle \Bigg \{  g''_2(\omega_n) W \, \frac{\log L}{\sqrt{L}} \,  + \sqrt{\frac{2}{P}}\underset{\omega}{\emph{sup}} \left \{\kappa_m(\omega) |\Psi_{m,r,t}(\omega)| \right\}  \Bigg\}^2$},
\end{split}
\label{eq50}
\end{equation}
where $g''_1(\omega)$, $g''_2(\omega)$ are bounded functions of $B$, $L$, $W$ and $\kappa_m(\cdot)$ is a function of $\omega$, given in Appendix \ref{Appendix_B}.
\end{corollary}
\begin{proof}
The proof is given in Appendix \ref{Appendix_B}.
\end{proof}

\section{Simulation Studies}\label{sec:sim}
First, we applying the PPMT-ESD estimator to simulated data and compare its performance to some of the existing methods. We consider two cases based on multivariate autoregressive processes: 

\emph{Case 1:} Estimating the spectral density matrix of a latent trivariate random process, given binary spiking observations.

\emph{Case 2:} Estimating the spectral density matrix of two processes, where one is directly observable and the other is latent and observed via binary spiking.

As for comparison, we consider two existing methods for extracting spectral representations of spiking data:

\subsubsection*{State-Space ESD Estimator}

This estimator closely follows the approach in \cite{Smith}. Recall that the observations are $n_{k}^{(l)} \sim \operatorname{Bernoulli}({\lambda}_k)$, where ${\lambda}_k = \operatorname{logistic} \, (x_k)$. Using the same notation as in \cite{Smith}, we model the latent process $x_k$ as a first-order autoregressive model, $x_k = x_{k-1} + \epsilon_k$, where $\epsilon_k \overset{i.i.d.}{\sim} \mathcal{N}\, (0, \sigma_{\epsilon}^2)$ for $1 \leq k \leq K$. Following an EM algorithm developed in \cite{Smith}, the MAP estimate of $x_k$ given the observed data, $x_{k|K}$, is obtained. Then, the direct multitaper PSD of the estimated process $x_{k|K}$ is taken as the PSD estimate of $x_k$. For the multivariate non-stationary case, for each process $x_{k,j}$, $1 \leq j \leq J$, we assume joint stationarity within non-overlapping consecutive windows of length $W$. The MAP estimate of each process $\mathbf{x}_{m,j} = [x_{(m-1)W + 1,j},  x_{(m-1)W + 2,j}, \cdots, x_{mW,j}]^\top$, for $1 \leq m \leq M$ and $1 \leq j \leq J$ are obtained, and an estimate of the ESD matrix is derived using the corresponding estimators. We refer to this estimator as the State-Space ESD (SS-ESD) estimator. 

\subsubsection*{Peristimulus Time Histogram ESD Estimator}

This estimator is derived by directly considering the ensemble mean of the spiking observations $\overline{n}_{k,j}$, referred to as the peristimulus time histogram (PSTH), as an estimate of the random signal $x_{k,j}$, for $1 \leq k \leq K$ and $1 \leq j \leq J$ \cite{Halliday1999}. With a similar joint stationarity assumption in windows of length $W$, the non-overlapping sliding window multitaper spectral estimate of the PSTH forms the PSTH-ESD estimator. 

In order to benchmark our comparison, we consider the theoretical spectra of the AR processes derived using closed-form expressions (True ESD) as well as the the non-overlapping sliding window direct multitaper estimates of the processes $x_{k,j}$ that have been used to generate the spikes. We refer to the later benchmark as the Oracle ESD, as if an oracle could directly observe the latent processes and estimate their ESD.

\vspace{-5mm}
\subsection{Case 1: Latent trivariate process observed through spiking}\label{sec:sim1}
\setcounter{figure}{0}
\begin{figure}[b!] 
	\centering
	\includegraphics[width=0.9\columnwidth]{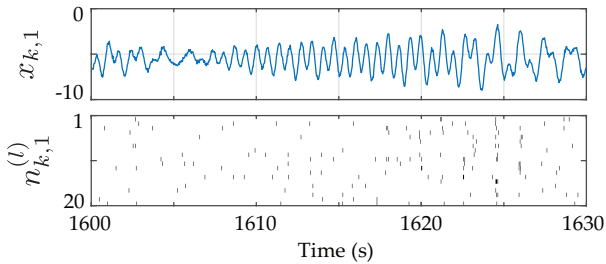}
	\caption{\small Samples of the signal $x_{k,1}$ (top) and the raster plot of the corresponding spikes (bottom) from $t = 1600$~s to $t = 1630$~s.}
	\label{[Fig: 1]}
\end{figure}

\setcounter{figure}{1}
\begin{figure*}[t!] 
	\centering
	\includegraphics[]{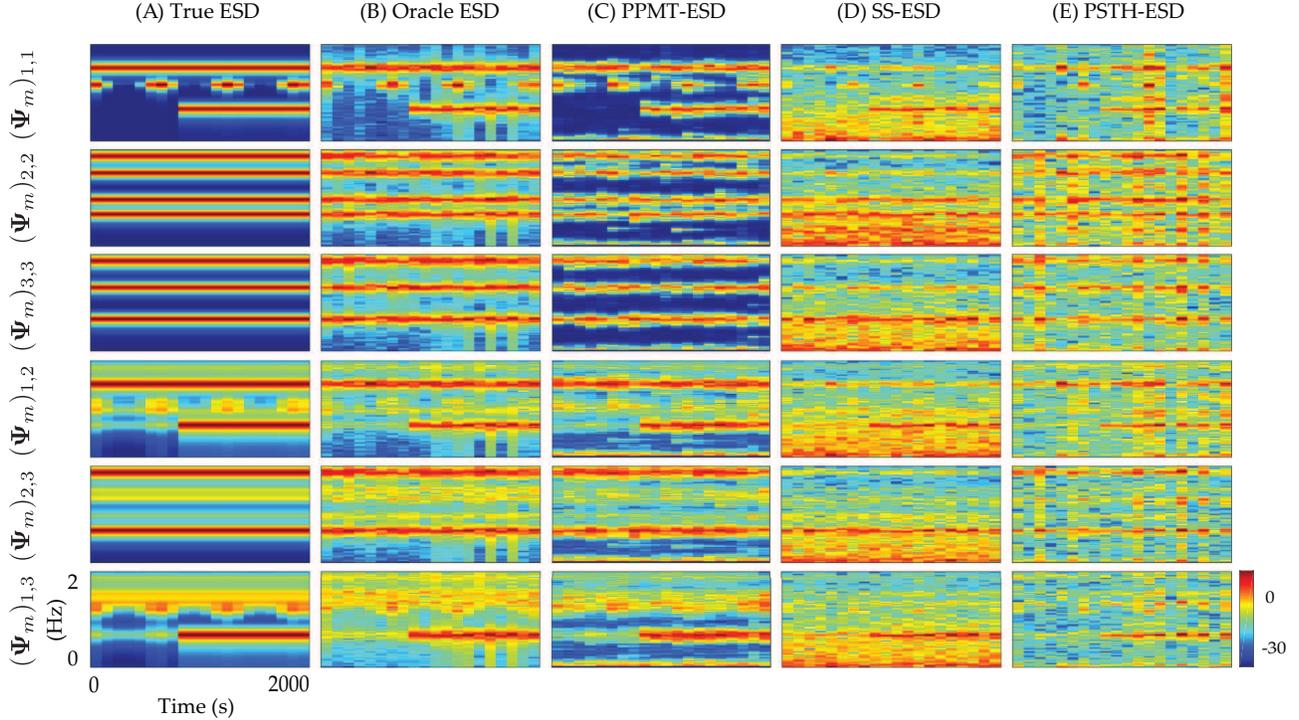}
	\vspace{-1mm}
	\caption{\small ESD estimation results for Case 1. Each panel shows the magnitude of the spectrogram in dB scale. Columns from left to right: $(A)$ True ESD, $(B)$ Oracle ESD, $(C)$ PPMT-ESD, $(D)$ SS-ESD, and $(E)$ PSTH-ESD. Rows from top to bottom: $(\bm{\Psi}_{m})_{1,1}(\omega)$, $(\bm{\Psi}_{m})_{2,2}(\omega)$, $(\bm{\Psi}_{m})_{3,3}(\omega)$, $(\bm{\Psi}_{m})_{1,2}(\omega)$,  $(\bm{\Psi}_{m})_{2,3}(\omega)$, and $(\bm{\Psi}_{m})_{1,3}(\omega)$.}
	\label{[Fig: 2]}
	\vspace{-4mm}
\end{figure*}

In order to simulate the three latent processes ($J=3$) with spectral couplings, we consider different linear combinations of a set of AR(6) processes, $\{ \{y_{k}^{(i)}\}_{k=1}^{K}, \, 1 \leq i \leq 6\}$, where $y_k^{(i)}$ is tuned around the frequency $f_i$, with $f_1 = 1.15 $ Hz, $f_2 = 0.95 $ Hz, $f_3 = 1.3 $ Hz, $f_4 = 1.5 $ Hz, $f_5 = 0.65 $ Hz and $f_6 = 1.85 $ Hz. All signals have been sampled at a sampling rate of $f_s = 32$ Hz, for a total duration of $2000$ seconds ($K = 64000$). The trivariate random process is defined as:
\begin{align}
\nonumber & \resizebox{\columnwidth}{!}{$\displaystyle x_{k,1} =  y_{k}^{(1)} \cos\left(2\pi {\textstyle \frac{f_0}{f_s}} k\right) + 1.2 y_{k}^{(4)} + 1.2 y_{k}^{(5)} u_{k - 0.4 K} + \sigma_{x1} \, \nu_{1,k} + x_{1,dc}$}\\
\nonumber &\resizebox{\columnwidth}{!}{$\displaystyle x_{k,2} =  0.83 \, y_{k}^{(2)} + 0.83 \, y_{k - 6}^{(4)} + 0.83 \, y_{k}^{(5)} + 0.83 \, y_{k}^{(6)}  + \sigma_{x2} \, \nu_{2,k}  + x_{2,dc}$} \\
\nonumber &\resizebox{\columnwidth}{!}{$\displaystyle x_{k,3} =  y_{k}^{(3)} + y_{k}^{(5)} + y_{k-10}^{(6)} \, u_{\, 0.5 K - 1 - k} + y_{k}^{(6)} \, u_{k - 0.5 K} + \sigma_{x3} \, \nu_{3,k}  + x_{3,dc}$}
\end{align}
where $x_{i,dc}$ are the DC components, $u_k$ is the unit step function, $\nu_{i,k}$ is a zero mean white Gaussian noise with unit variance, and $\sigma_{i}$ is a scaling standard deviation to set the SNR of all signals at $20$~dB, for $i=1,2$ and $3$.

In words, $x_{k,1}$ has been formed by combining three AR components, $y_{k}^{(1)}$, $y_{k}^{(2)}$ and  $y_{k}^{(5)}$. The component related to $y_{k}^{(1)}$ has been amplitude modulated by a low frequency cosine signal at $f_0 = 0.0008$ Hz. The signal initially consists of only  $y_{k}^{(1)}$ and  $y_{k}^{(4)}$, and the third component $y_{k}^{(5)}$ is added to $x_{k,1}$ after $800$ seconds. The process $x_{k,2}$ is composed of four AR components, $y_{k}^{(2)}$, $y_{k}^{(4)}$, $y_{k}^{(5)}$ and $y_{k}^{(6)}$, with the contribution from $y_{k}^{(4)}$ having a lag of $6$ samples. The third process $x_{k,3}$ consists of the three AR components of $y_{k}^{(3)}$, $y_{k}^{(5)}$ and $y_{k}^{(6)}$. Although $y_{k}^{(6)}$ appears with a lag of $10$ samples initially, it becomes in phase with after $1000$ seconds.

 We generate spike trains for $L = 20$ neurons per each of the latent processes, using the logistic link model of Eq. (\ref{eq:logistic}). All DC components in the CIF have been set to $-5.5$, so that the average spiking rate of the ensemble corresponding to each CIF is $\approx 0.28$ spikes/second, consistent with the low spiking rate of experimentally recorded data. A $30$ second sample window of the process $x_{k,1}$ and the corresponding raster plot of the spiking observations is shown in Fig. \ref{[Fig: 1]}.

We take the window length of stationary to be $100$ seconds ($W = 3200$), resulting in a total number of $M=20$ windows.  The parameter $\alpha$ has been fixed at $0.4$ to have optimal dependency across time windows and the prior parameter $\rho$ has been set to $0.2$. Note that these hyper-parameters can be further fine tuned using cross validation, in order to achieve a higher precision. Furthermore, we set $N = 800$ in order to have a densely sampled spectral representation. Given that the spectral range is known to be $[0,2]$~Hz, we only use the first $N_{\max} = 100$ frequency bins of the matrix $A$ for the sake of computational complexity. Finally, the time-bandwidth product of the multitaper framework is chosen as $2$ ($\xi = 2$), and the first three dpss tapers are used.

Fig. \ref{[Fig: 2]} shows the estimation results corresponding to this case. The results have been formatted as a grid, with the columns representing (from left to right) the True EDS, Oracle ESD, PPMT-ESD, SS-ESD, and PSTH ESD estimates. The rows represent (from top to bottom) $(\bm{\Psi}_{m})_{1,1}(\omega)$, $(\bm{\Psi}_{m})_{2,2}(\omega)$, $(\bm{\Psi}_{m})_{3,3}(\omega)$, $(\bm{\Psi}_{m})_{1,2}(\omega)$, $(\bm{\Psi}_{m})_{2,3}(\omega)$ and $(\bm{\Psi}_{m})_{1,3}(\omega)$, where $(\bm{\Psi}_{m})_{i,j}(\omega)$ denotes the magnitude of the $(i,j)^{\sf th}$ block of the ESD matrix. Moreover, in order to have a closer inspection, the magnitude of the spectra corresponding to a window of $t = 700$~s to $t = 800$~s for $(\bm{\Psi}_{m})_{1,1}(\omega)$, $(\bm{\Psi}_{m})_{2,2}(\omega)$ and $(\bm{\Psi}_{m})_{1,2}(\omega)$ are depicted in Fig. \ref{[Fig: 3]}. 

\setcounter{figure}{2}
\begin{figure}[t!] 
\vspace{-2mm}
	\centering
	\includegraphics[]{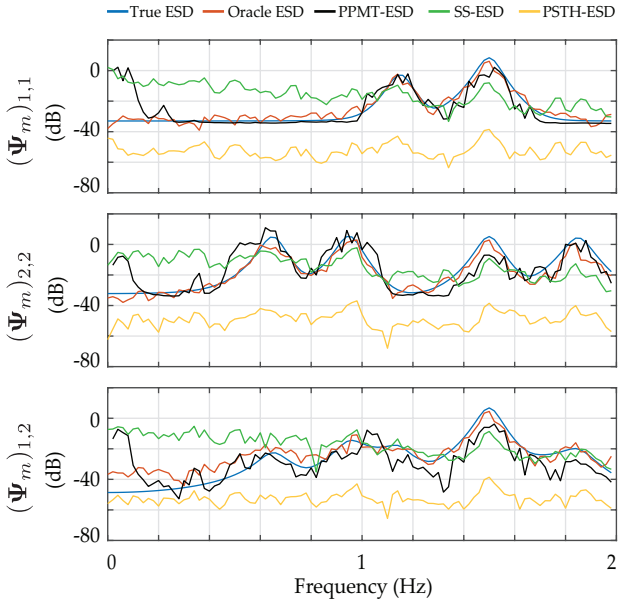}
	\vspace{-2mm}
	\caption{\small A snapshot of the spectrograms of Fig. \ref{[Fig: 2]} at the $8^{\sf th}$ window ($t = 700$~s--$800$~s). Rows from top to bottom: $(\bm{\Psi}_{m})_{1,1}(\omega)$, $(\bm{\Psi}_{m})_{2,2}(\omega)$, and $(\bm{\Psi}_{m})_{1,2}(\omega)$.}
	\label{[Fig: 3]}
	\vspace{1mm}
\end{figure}

It can be observed that the proposed PPMT-ESD estimator (Fig.  \ref{[Fig: 2]}(C)) results in much less background noise compared to all the others, while precisely capturing the dynamic evolution of the spectra and properly resolving the various frequency components. The latter is more evident from Fig. \ref{[Fig: 3]}, where the PPMT-ESD (black trace) closely matches the true ESD (blue trace) on par with the Oracle ESD (red trace), while the SS-ESD (green trace) and PSTH-ESD (orange trace) show significant bias and variability. It is worth mentioning that the erroneous spectral peak near the DC component in Fig. \ref{[Fig: 3]} (black trace) is due to the estimation error of the DC component in the CIF model, in the low spiking regime of our setting. Nevertheless, this peak appears in the SS-ESD and PSTH-ESD estimates as well, but is mitigated by increasing the spiking rate. 

Due to the time-domain smoothing model used in the SS-ESD estimator, the ESD rapidly decays with frequency (Fig. \ref{[Fig: 3]}, green trace). As such, the SS-ESD estimate (Fig. \ref{[Fig: 2]}(D)) heavily amplifies non-existing low frequency components that arise from the intrinsic noise in spiking observations, while suppressing the higher frequency components that exist in the true ESD (Fig. \ref{[Fig: 2]}(A)). Similar, the estimate of the PSTH-ESD shown in Fig. \ref{[Fig: 2]}(E) fails to capture most of the temporal and spectral features of the ESD, since it does not account for the binary nature of the observations. To quantify the performance of these estimators, we repeated this numerical experiment for a total of $50$ trials, generating independent realizations of the AR processes and spiking observations per trial, and computed the Mean Squared Error (MSE) with respect to the the True ESD (in dB scale). The average and variance of the MSE values are presented in Table \ref{tab1}. All MSE computations have been normalized with respect to the total power of the True ESD (in dB scale). As expected, the Oracle ESD achieves the lowest MSE. Among the three methods that use the spiking observations, our proposed PPMT-ESD estimator achieves the lowest MSE, followed by the SS-ESD estimator with a significant gap. The PSTH-ESD estimator exhibits the poorest performance, in terms of both average MSE and variance, which is also visually evident in Fig. \ref{[Fig: 3]}. In the spirit of easing reproducibility, a MATLAB implementation that regenerates the data, results and figures outlined in this section has been made publicly available on the open source repository GitHub \cite{GitHub_repository}.

\vspace{-2mm}
\begin{figure*}[t!] 
	\centering
	\includegraphics[]{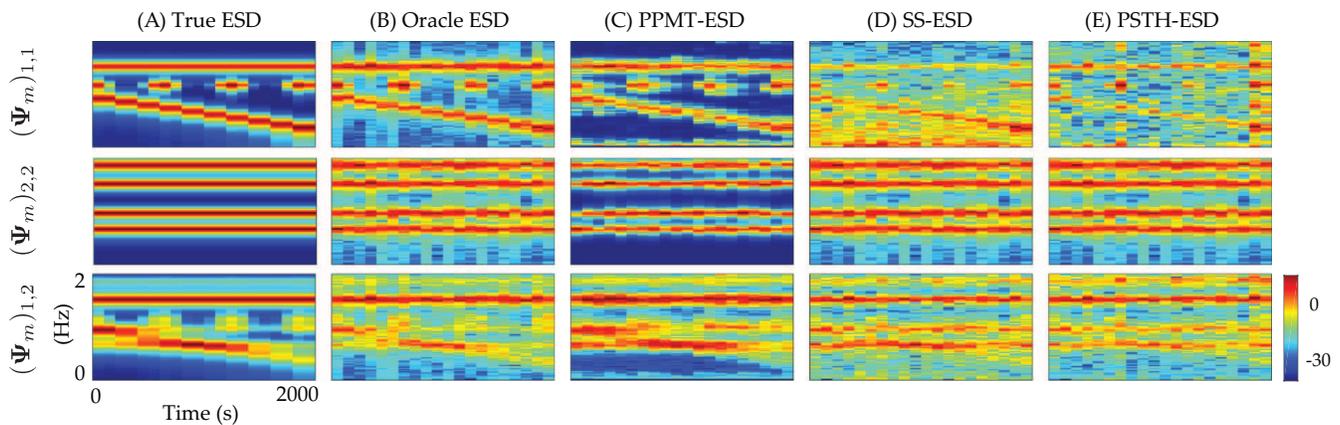}
	\caption{\small ESD estimation results for Case 2. Each panel shows the magnitude of the spectrogram in dB scale. Columns from left to right: $(A)$ True ESD, $(B)$ Oracle ESD, $(C)$ PPMT-ESD, $(D)$ SS-ESD, and $(E)$ PSTH-ESD. Rows from top to bottom: $(\bm{\Psi}_{m})_{1,1}(\omega)$, $(\bm{\Psi}_{m})_{2,2}(\omega)$, and $(\bm{\Psi}_{m})_{1,2}(\omega)$.}
	\label{[Fig: 4]}
	\vspace{-4mm}
\end{figure*}

\begin{table} [!ht]
\vspace{-2mm}
\caption{\small Comparison of Relative MSE Performance} 
\label{tab1}
\centering
\resizebox{0.9\columnwidth}{!}{\begin{tabular}{ | c | c | c |} 
\hline
Estimation method & Average MSE & Variance of MSE \\ 
\hline
Oracle ESD  & 0.0485 & $8.2612 \times 10^{-6}$ \\ 
PPMT-ESD  & 0.1864 & $7.4279 \times 10^{-5}$\\ 
SS-ESD  & 0.3911 & $1.3991 \times 10^{-5}$\\ 
PSTH-ESD  & 1.4794 & $1.1 \times 10^{-3}$\\ 
\hline
\end{tabular}}
\vspace{-4mm}
\end{table}

\vspace{-2mm}
\subsection{Case 2: Latent bivariate process with one directly observable component}\label{sec:sim2}

While Case 1 was a natural choice for performance comparison, Case 2 is of particular interest in the joint analysis of neural spiking and continuous signals, such as the local field potential (LFP). The LFP signal corresponds to the electrical field potential measured at the cortical surface, and mesoscale dynamics of cortical activity.

In this case, we consider a bivariate random process, whose first component $x_{k,1}$ is observed through spiking activity $\{ n_{k,1}^{(l)} \}_{k,l = 1}^{K, L}$, while its second component $x_{k,2}$ is directly observable in i.i.d. zero-mean Gaussian noise, i.e., $\widetilde{x}_{k,2} := x_{k,2} + \nu_k$, with $\nu_k \sim \mathcal{N}\, (0, \sigma_{\nu}^2)$. Using the same notations as before, the two processes used in this simulation are,
\begin{equation}
\begin{split}
 \nonumber & \resizebox{0.83\columnwidth}{!}{$\displaystyle x_{k,1} = y_{k}^{(1)} \cos\left(2\pi {\textstyle \frac{f_0}{f_s}} k\right) + y_{k}^{(4)} + y_{k}^{(7)} + \sigma_{x1}  \nu_{1,k} + x_{1,dc}$} \\
 & \resizebox{\columnwidth}{!}{$\displaystyle x_{k,2} = 0.83 \, y_{k}^{(2)} + 0.83 \, y_{k - 6}^{(4)} + 0.83 \, y_{k}^{(5)} + 0.83 y_{k}^{(6)} + \sigma_{x2} \nu_{2,k}  + x_{2,dc}.$}
\end{split}
\label{eq52}
\end{equation}

The process $x_{k,2}$ considered in this scenario is exactly the same as that described in Case 1. Process $x_{k,1}$ is a slightly modified version of the process $x_{k,1}$ in Case 1: the first two components of $x_{k,1}$ are the same as described before; the amplitude modulated component related to $y_k^{(1)}$ and the fixed frequency AR component $y_k^{(4)}$. However, instead of the third component $y_k^{(5)}$, we include a frequency modulated component, $y_k^{(7)}$, which is an AR process tuned around the frequency $f_7$. The frequency $f_7$ changes by decrements of $0.06$~Hz every $200$ seconds, starting at $0.9$~Hz at $t = 0$~s.

The ESD matrix of this bivariate process can be estimated by Algorithm 2 with a slight modification in the Forward filtering step (step 2) of Algorithm 1: Given that the second process is directly observable, the distribution $f(\mathcal{D}_1^m | \{ \mathbf{V} \}_1^m , \widehat{\bm{\theta}}^{(r)})$ needs to be modified, and accordingly, the log-posterior in Eq. (\ref{eq14}) changes to:  
\begin{equation}
\begin{split}
\nonumber & \resizebox{\columnwidth}{!}{$\displaystyle \sum_{s,w = 1}^{m,W}  L \,\Big\{\overline{n}_{(s-1)W+w,1} (\mathbf{A}_s \mathbf{V}_s)_{w,1}   \, - \log \, (1 + \exp(\mathbf{A}_s \mathbf{V}_s)_{w,1}))\Big\}$}\\
\nonumber  & \resizebox{\columnwidth}{!}{$\displaystyle - \sum_{s, w = 1}^{m, W} \frac{1}{2 \sigma_{\nu}^2} \Big( \widetilde{x}_{(s-1)W+w,2} \!-\! (\mathbf{A}_s \mathbf{V}_s)_{w,2} \Big)^2 \!\!\!-\! \frac{1}{2} \, \sum_{s = 1}^{m} (\mathbf{w}_{s} \!-\!  \mathbf{\Phi} \mathbf{w}_{s-1})^T (\mathbf{Q}_s^{(r)})^{-1} (\mathbf{w}_{s} \!-\!  \mathbf{\Phi} \mathbf{w}_{s-1}) \Bigg).$}
\end{split}
\label{eq53}
\end{equation}
The parameters used in forming the estimates are the same as those chosen for Case 1. Fig. \ref{[Fig: 4]} shows the estimation results corresponding to this case. The results have been similarly formatted as a grid, with the columns representing (from left to right) the True EDS, Oracle ESD, PPMT-ESD, SS-ESD, and PSTH ESD estimates. The rows represent (from top to bottom) $(\bm{\Psi}_{m})_{1,1}(\omega)$, $(\bm{\Psi}_{m})_{2,2}(\omega)$, and $(\bm{\Psi}_{m})_{1,2}(\omega)$. Note that in this case, we take the SS-ESD and PSTH-ESD estimates of $(\bm{\Psi}_{m})_{2,2}(\omega)$ to be the same as its Oracle ESD estimate, given that these methods are based on estimating the process $x_{k,2}$ in time domain (which is directly observable here). 

Similar to Case 1, the proposed PPMT-ESD estimator (Fig. \ref{[Fig: 4]}(C)) captures the dynamics of the spectra $(\bm{\Psi}_{m})_{1,1}(\omega)$ and $(\bm{\Psi}_{m})_{1,2}(\omega)$ accurately, closely matching the True ESD (Fig. \ref{[Fig: 4]}(A)). As before, the SS-ESD estimator (Fig. \ref{[Fig: 4]}(D)) is not able to capture the ESD dynamics, especially at high frequencies. Though some frequency components at certain time windows are detected by the PSTH-ESD estimates (Fig. \ref{[Fig: 4]}(E)), most of the frequency content is concealed by background noise.

\vspace{-2mm}
\section{Application to experimentally-recorded data from general anesthesia}\label{sec:data}

Finally, we apply our proposed PPMT-ESD estimator to multi-unit recordings from a human subject under Propofol-induced general anesthesia (data from \cite{lewis2013local,Miran}). The data set includes the spiking activity of $41$ neurons as well as the LFP recorded from a patient undergoing intra-cranial monitoring for surgical treatment of epilepsy using a multichannel micro-electrode array implanted in temporal cortex \cite{lewis2013local}. Recordings were conducted during the administration of Propofol for induction of anesthesia. The experimental protocol is extensively explained in \cite{lewis2013local}. Since the original recordings have been over-sampled, we down-sample both the LFP and spike recordings to the same sampling frequency of $25$~Hz \cite{Miran}. 

From the spike recordings, we select $25$ neurons ($L = 25$) with the highest spiking rates for analysis. The average spiking rate of this subpopulation of neurons is $0.1732$ spikes/second. Fig. \ref{[Fig: 5]}(a) shows the LFP signal for a total duration of $800$~s considered in the analysis. A zoomed-in view of the LFP and the raster plot of the neuronal ensemble corresponding to a $40$~s window from $t = 480$~s to $t = 520$~s is shown in Fig. \ref{[Fig: 5]}. 

\begin{figure}[!ht] 
	\centering
	\vspace{-2mm}
	\includegraphics[width=0.9\columnwidth]{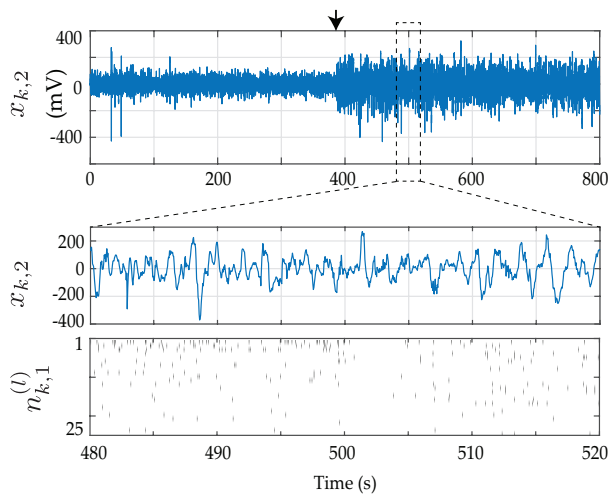}
	\vspace{-1mm}
	\caption{\small Top panel shows the LFP recording used in the analysis. The downward arrow marks the induction of the anesthetic. The bottom rows show the zoomed-in view of the LFP and the spike trains for $t = 480$~s to $t = 520$~s.}
	\label{[Fig: 5]}
\end{figure}

Similar to Case 2 in Section \ref{sec:sim2}, we model the neuronal ensemble and the LFP by a bivariate process, where the LFP is directly observable, but the latent process driving the neuronal activity is observed through spiking. We assume the process to be stationary within windows of length $40$~s ($W = 1000$), resulting in $M=20$ non-overlapping windows. Given that the main spectral content is known to be in the range of $0$--$2$ Hz \cite{lewis2013local,Miran}, we estimate the ESDs up to $2.5$~Hz ($N_{\max} = 100$). Further, we choose $\xi=2$, and use the first three tapers in the analysis. We set $\alpha = 0.85$ and $\rho = 0.02$.     

Fig. \ref{[Fig: 6]} shows the results of our analysis. The first column illustrates the ESD estimates, where the dashed vertical line marks the induction of anesthetic. The second column shows snapshots of the normalized magnitude spectrum corresponding to a $40$~s window from $t = 480$~s to $t = 520$~s. The first three rows (from top to bottom) show the estimates of $(\bm{\Psi}_{m})_{1,1}(\omega)$, i.e., the ESD of the latent process driving spiking activity, for PSTH-ESD, SS-PSD, and the proposed PPMT-ESD, respectively. The fourth row represents the Oracle ESD estimate of the LFP signal ($(\bm{\Psi}_{m})_{1,1}(\omega)$), and the fifth row depicts ($(\bm{\Psi}_{m})_{1,2}(\omega)$), the cross-spectra between the LFP and the latent process, using the proposed PPMT-ESD method.

\begin{figure}[!ht] 
	\centering
	\includegraphics[width=0.95\columnwidth]{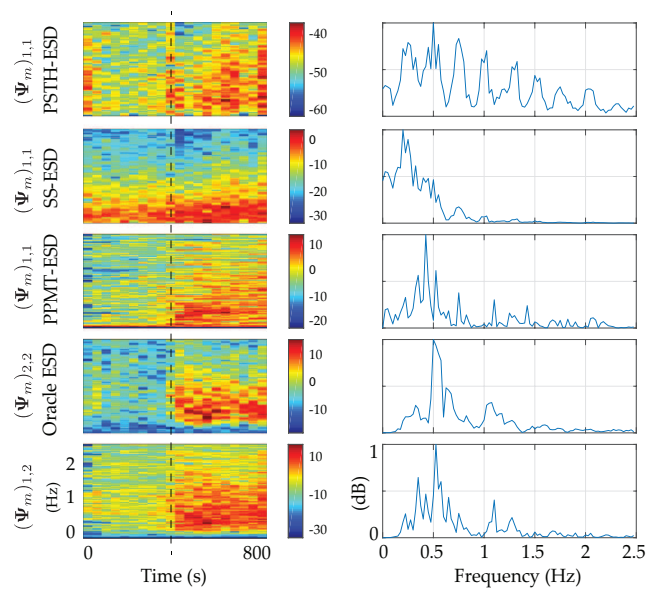}
	\caption{\small ESD analysis of the real data from anesthesia. Left column shows the magnitude of the spectrograms in dB scale, and the right column shows a snapshot of the normalized ESD estimates for a $40$~s window starting at $t = 480$~s. Rows from top to bottom: PSTH-ESD, SS-ESD, PPMT-ESD, the Oracle ESD of the LFP signal, and the PPMT-ESD estimate of the cross-spectrum between the spiking observations and the LFP.}
	\label{[Fig: 6]}
\end{figure}

Similar to our simulation studies, we observe that the PSTH-ESD estimate (Fig. \ref{[Fig: 6]}, first row) captures a significant number of spurious frequency components and harmonics, which are known to be absent during Propofol general anesthesia \cite{lewis2013local}, masking the relevant spectral content. Also, as before, the SS-ESD estimate amplifies the low-frequency components (Fig. \ref{[Fig: 6]}, second row). Even though the dominant spectral content of the LFP is around $0.5$~Hz and $1.1$~Hz, as observed in the right panel of the fourth row in Fig. \ref{[Fig: 6]}, neither of these peaks are present in the SS-ESD estimate (Fig. \ref{[Fig: 6]}, second row, right panel). In addition, the rapid change in the LFP spectrogram (fourth row, left panel) marked by the dashed line, is not captured by the SS-ESD estimate (second row, left panel).

The proposed PPMT-ESD estimate (third row) closely resembles the spectrum of the LFP (fourth row), evident in both the spectrogram and the spectral snapshot. The dominant frequency components of the PMTM-ESD around $0.5$~Hz and $1.1$~Hz match those in the LFP spectrum. The temporal variations of the spectrum are also consistent with that of the LFP, as visible in the left panels. The estimated cross-spectrum (fifth row) further corroborates this observation, by capturing the spectral coupling between the LFP and the latent process driving the spiking activity. It is indeed hypothesized that the latent process driving the spiking activity is the LFP signal itself \cite{lewis2013local}. As such, our analysis corroborates this hypothesis by providing an accurate estimate of the ESDs of both the LFP and the latent process driving the neuronal spiking. 

\vspace{-2mm}
\section{Concluding Remarks}\label{sec:con}
\vspace{-1mm}
Brain oscillations are known to play a significant role in modulating the finescale dynamics of neuronal spiking. These oscillations are also known to be non-stationary, as they reflect the changes in the internal states of behavioral conditions. Understanding how the latent processes that drive spiking activity are related to brain oscillations is a key problem in computational neuroscience. On one hand, the theory of point processes has been successfully utilized in recent years to capture the dynamics of neuronal spiking data. On the other hand, spectral estimation of non-stationary continuous signals has been widely studied. Existing methods either treat the spiking data as continuous, and apply non-stationary spectral estimation techniques off-the-shelf, or try to first estimate these latent processes in time-domain, followed by forming spectrograms. Both these approaches are known to suffer from high bias and variability. A unified methodology for inferring spectrotemporal representations of the multivariate latent non-stationary processes that drive neuronal spiking is lacking. In this work, we proposed such a methodology for estimating the ESD matrix of a multivariate non-stationary latent process \emph{directly} from binary spiking observations. To this end, we integrated techniques from state-space modeling, multitaper analysis and point processes. 

We established theoretical bounds on the bias and variance performance of the proposed estimator, and compared its performance with the aforementioned existing techniques through application to simulated and experimentally-recorded neural data. Our simulation studies confirmed our theoretical analysis and revealed the favorable performance of our proposed method over existing approaches. Our application to real data corroborated the hypothesis on the role of the local field potential in regulating spiking activity under general anesthesia, by providing a clear picture of the underlying spectral couplings. While we have developed our methodology in the context of neuronal data analysis, it can be applied to a wide range of discrete observations that are modulated by underlying oscillations, such as heart-beat  data \cite{barbieri2005point}. Our methodology can also be extended to infer non-stationary network-level properties such as the frequency domain Granger-Geweke causality \cite{geweke1982measurement,baccala2001partial}. 

\vspace{-2mm}
\section{Acknowledgment}
We would like to thank Emery N. Brown and Patrick L. Purdon for sharing the data from \cite{lewis2013local}.

\appendices
\vspace{-2mm}
 \section{Proof of Theorem \ref{thm1}}
\label{Appendix_A}
\begin{proof}
Let $S(\omega)$ be the PSD of the process $\{x_k\}_{k=1}^K$. Then,
\begin{align}
\begin{split}
& |\operatorname{bias}(\widehat{S}^{\sf mt}(\omega))| \; := \; |\mathbb{E}[\widehat{S}^{\sf mt}(\omega)] - S(\omega)| \\
&  \overset{(a)}{\leq} \, |\mathbb{E}[\widehat{S}^{\sf mt}(\omega) \, - \, S^{\sf mt}(\omega)]| \,  + \, |\operatorname{bias}(S^{\sf mt}(\omega))|, 
\end{split}
\label{eq26}
\end{align}
where $(a)$ follows from the triangle inequality. Further,
\begin{align}
\begin{split}
&  \operatorname{Var}(\widehat{S}^{\sf mt}(\omega)) \; := \; \mathbb{E}[|\widehat{S}^{\sf mt}(\omega) - \mathbb{E}[\widehat{S}^{\sf mt}(\omega)]|^2]  \\  
& \leq \,  \mathbb{E}[|\widehat{S}^{\sf mt}(\omega) \, - \, S^{\sf mt}(\omega)|^2] + \operatorname{Var}(S^{\sf mt}(\omega))  \\
&  +  2  \mathbb{E}[(\widehat{S}^{\sf mt}(\omega) \, - \, S^{\sf mt}(\omega))(S^{\sf mt}(\omega) \,-\, \mathbb{E}[S^{\sf mt}(\omega)])] \\ 
&  \resizebox{0.85\columnwidth}{!}{$\displaystyle \overset{(b)}{\leq} \left\{ \sqrt{\mathbb{E}[|\widehat{S}^{\sf mt}(\omega) \, - \, S^{\sf mt}(\omega)|^2]} \, + \, \sqrt{\operatorname{Var}(S^{\sf mt}(\omega))} \right\}^2$},
 \end{split}
 \label{eq27}
\end{align}
where $(b)$ follows from the Cauchy-Schwarz inequality. Thus, the desired bounds on the bias and variance can be established through bounding the first and second moments of $(\widehat{S}^{\sf mt}(\omega) \, - \, S^{\sf mt}(\omega) )$. The first moment can be bounded by
\vspace{-3mm}
\begin{align}
\nonumber |\mathbb{E}&[\widehat{S}^{\sf mt}(\omega) -  S^{\sf mt}(\omega)]| \overset{(c)}{\leq}  \, \frac{1}{P} \sum_{p = 1}^{P} \left|\mathbb{E} \left[ |\widehat{y}^{(p)}(\omega)|^2  -  |y^{(p)}(\omega)|^2 \right]\right|
\end{align}

\begin{align}
\nonumber \overset{(d)}{\leq} 
& \frac{1}{P}\! \sum_{p = 1}^{P}\! \sum_{k = 1}^{K}\! \sum_{m = 1}^{K} \! \resizebox{0.715\columnwidth}{!}{$\displaystyle \left| \mathbb{E} [ \operatorname{logit} (\overline{n}_{k}) \operatorname{logit} (\overline{n}_{m}) \!-\!  x_k x_m ]  \nu_k^{(p)} \nu_m^{(p)}\, e^{-i \omega (m-k)} \right|$} \\
\nonumber \overset{(e)}{\leq} 
& \frac{1}{P} \sum_{p = 1}^{P} \sum_{k = 1}^{K} \left| \mathbb{E} \left[ (\operatorname{logit} (\overline{n}_{k}))^2 \, - x_k^2 \right]  \right|(\nu_k^{(p)})^2 \\ 
& \!\!\!\!\!\!\!+\! \!\frac{1}{P}\! \sum_{p = 1}^{P} \! \sum_{k \neq m}^{} \!\! \left| \mathbb{E} \left[ \operatorname{logit} (\overline{n}_{k}) \! \operatorname{logit} (\overline{n}_{m}) \!-\! x_k x_m\right] \right| \!\left| \nu_k^{(p)}\! \nu_m^{(p)} \right|\!,
\label{eq28}
\end{align}
where $(c)$ and $(d)$ follow from the triangle inequality and $(e)$ follows by bounding the complex sinusoid. The main technical difficulty in further development of the bounds is due to the fact that $\operatorname{logit}(z)$ does not have a Taylor series expansion for $z \in (0,1)$. We thus have to find other algebraically useful bounds. To this end, we need to establish the following technical lemma.

\begin{lemma} Consider the event $ A =  \{\overline{n}_{k} \; | \; \overline{n}_{k} \neq 0, \, \overline{n}_{k} \neq 1, \, 1 \leq k \leq K \} $. The following inequality holds true for all $\overline{n}_{k} \in A $:
\begin{align}
 \nonumber \varepsilon(\overline{n}_{k}) := \left| \operatorname{logit} (\overline{n}_{k}) - x_k \right| \; \leq \; g(x_k, L) \, \left| \overline{n}_{k} - \lambda_k \right|,
 \end{align}
 where
 \vspace{-2mm}
 \begin{align} 
  \nonumber g(x_k, L) = \max \resizebox{0.7\columnwidth}{!}{$\displaystyle \left\{ \frac{1}{\lambda_k (1 - \lambda_k)}, \, \frac{|\log(L-1) + x_k|}{|\lambda_k - 1/L|}, \frac{|\log(L-1) - x_k|}{|1 - 1/L - \lambda_k|} \right\}$}.
\end{align}
\label{eq29}
\label{lemma1}
\end{lemma}

\vspace{-7mm}
\begin{proof}[Proof of Lemma \ref{lemma1}]
First, consider the case $ \lambda_k \leq 0.5$. We bound the function $\varepsilon(\overline{n}_{k})$ in a piece-wise fashion as follows. Note that $\operatorname{logit} (\overline{n}_{k})$ is convex for $\overline{n}_{k} \geq 0.5$ and concave for $\overline{n}_{k} \leq 0.5$. Thus, it immediately follows that for $\overline{n}_{k} \leq \lambda_k$, $\varepsilon(\overline{n}_{k})$ is convex and hence,
\begin{equation}
\begin{split}
\varepsilon(\overline{n}_{k}) \quad \leq \quad \frac{|\log(L-1) + x_k|}{|\lambda_k - 1/L|} \, \left(\lambda_k  - \overline{n}_{k}\right).
\end{split}
\label{eq30}
\end{equation}

Furthermore, for $\lambda_k \leq \overline{n}_{k} \leq 0.5$, $\varepsilon(\overline{n}_{k})$ is concave, and hence is bounded by the tangent at $\lambda_k$ as follows.
\begin{equation}
\begin{split}
\varepsilon(\overline{n}_{k}) \quad \leq \quad \frac{1}{\lambda_k (1 - \lambda_k)} \, \left(\overline{n}_{k} - \lambda_k \right)
\end{split}
\label{eq31}
\end{equation}

Finally, for the case of $\overline{n}_{k} \geq 0.5$, consider the line
\begin{equation}
\begin{split}
\ell(\overline{n}_{k}) := \frac{|\log(L-1) - x_k|}{|1 - 1/L - \lambda_k|} \, (\overline{n}_{k} - \lambda_k ).
\end{split}
\label{eq32}
\end{equation}
From the convexity of $\varepsilon(\overline{n}_{k})$, $\ell(\overline{n}_{k})$ upper bounds $\varepsilon(\overline{n}_{k})$ for $\overline{n}_{k} \geq 0.5$, since $\ell(0.5) \geq \varepsilon(0.5)$ for $ \lambda_k \leq 0.5$. Combining the piece-wise bounds in Eqs. (\ref{eq30}), (\ref{eq31}) and (\ref{eq32}), we conclude the claim in Lemma \ref{lemma1} for $\lambda_k \le 0.5$. Due to the symmetry of $\varepsilon(\overline{n}_{k})$, through a similar argument, the bound can be established for $\lambda_k > 0.5$, which concludes the proof.
\end{proof}
\vspace{-2mm}

Given that $|x_k| \; \leq \, B$ and assuming that $L$ is large enough so that $L \geq 2(1 + \exp(B))$, we can further simplify the bound of Lemma \ref{lemma1}. We have:
\begin{equation}
\nonumber \begin{split}
& g(x_k, L) \leq \max \resizebox{0.73\columnwidth}{!}{$\displaystyle \left\{ \exp(B)\, (1 + \exp(-B))^2, \frac{|\log(L-1) + B|}{(1 / (1 + \exp(B)) - 1/L)} \right\}$} \\ 
& \leq  \max \Big\{ \exp(B)\, (1 + \exp(-B))^2,  4(1 + \exp(B)) \log L \Big\}.
\end{split}
\label{eq33}
\end{equation}
Thus, for sufficiently large $L$, we conclude that, 
\begin{equation}
\begin{split}
\varepsilon(\overline{n}_{k}) \leq 4(1 + \exp(B)) \log L \left| \overline{n}_{k} - \lambda_k \right|.
\end{split}
\label{eq34}
\end{equation}
Now, consider the expectations in the bounds of Eq. (\ref{eq28}). Using iterated conditioning, 
\begin{equation}
\begin{split}
&\left| \mathbb{E} [ (\operatorname{logit} (\overline{n}_{k}))^2 \, - x_k^2 ]  \right|
=  \left| \mathbb{E}[ \mathbb{E}[(\operatorname{logit} (\overline{n}_{k}))^2 | x_k] \, - x_k^2 ]   \right| \\
&  = \left|\mathbb{E}[ 2 x_k \mathbb{E}[(\operatorname{logit} (\overline{n}_{k}) - x_k) | x_k]+ \mathbb{E}[(\operatorname{logit} (\overline{n}_{k}) \!-\! x_k)^2 | x_k]  ] \right| \\ 
&\overset{(f)}{\leq} \!\!\mathbb{E}[ 2|x_k| \mathbb{E}[ |\operatorname{logit} (\overline{n}_{k}) \!-\! x_k| | x_k]] \!+\! \mathbb{E}[\mathbb{E}[(\operatorname{logit} (\overline{n}_{k}) \!-\! x_k)^2 | x_k]],
\end{split}
\label{eq35}
\end{equation}
where $(f)$ follows from triangle and Jensen's inequalities. In order further simplify these bounds, we invoke the result of Lemma \ref{lemma1}. First, note that $\operatorname{logit} (\overline{n}_{k})$ is unbounded in the complement of event $A$. Provided $|x_k| \, \leq \, B$, $\mathbb{P}(\overline{n}_{k} \neq 0)$ and $\mathbb{P}(\overline{n}_{k} \neq 1)$ can be lower bounded by $1 - \exp(-L \log(1 + \exp(-B)))$, which implies that 
\vspace{-2mm}
\begin{equation}
\begin{split}
\nonumber \mathbb{P}(A) \quad \geq \quad 1 - 2 \exp(-L \log(1 + \exp(-B))).
\end{split}
\label{eq36}
\end{equation}
Therefore, for sufficiently large $L$, we see that $\mathbb{P}(A)$ is exponentially close to $1$. Thus,  hereafter we condition the expectations on the highly probable event $A$. From Eq. (\ref{eq34}), we get
\vspace{-2mm}
\begin{equation}
\begin{split}
\nonumber \mathbb{E}[\, |\operatorname{logit}( \overline{n}_{k}) - x_k| \,| \, x_k, A] \,  \leq \, &4(1 + \exp(B))\log L  \\
& \mathbb{E}[\, \left| \overline{n}_{k} - \lambda_k \right| \, | \,  x_k, A ].
\end{split}
\label{eq37}
\end{equation}
Note that the random variable $n_k = L \overline{n}_{k}$ is the sum of $L$ independent Bernoulli random variables given $x_k$. Thus, given $x_k$, $n_k \sim \operatorname{Binomial} (L, \lambda_k)$. Accordingly, 
\begin{align}
\nonumber \mathbb{E}[\, \left| \overline{n}_{k} - \lambda_k \right| \, | \,  x_k, A ] &= \mathbb{E}[\, \left| \overline{n}_{k} - \lambda_k \right| \mathbbm{1}_A\, | \,  x_k] / \, \mathbb{P}(A)\\ 
\nonumber & \overset{(g)}{\leq} \quad \sqrt{\mathbb{E}[\, ( \overline{n}_{k} - \lambda_k )^2 \mathbbm{1}_A\, | \,  x_k] }\, / \, \mathbb{P}(A)  \\
& \overset{(h)}{\leq}  \quad \sqrt{\lambda_k (1 - \lambda_k)} / (\sqrt{L}\mathbb{P}(A)), 
\label{eq38}
\end{align}
where $(g)$ follows from the Jensen's inequality and $(h)$ follows from substituting expression for the variance of a binomial random variable. Further, note that $\lambda_k (1 - \lambda_k) \leq 1/4$,  for $\lambda_k \in [0, 1]$ and $ \mathbb{P}(A) = 1 - (\lambda_k^L + (1 - \lambda_k)^L) \geq 1/2$, if $(1 + \exp(B))^L > 2(1 + \exp(BL))$, which is satisfied for large enough $L$. Thus, combining the bounds in Eqs. (\ref{eq38}) and (\ref{eq34}), we get,  
\begin{equation}
\begin{split}
\mathbb{E}[\, |\operatorname{logit}( \overline{n}_{k}) - x_k| \,| \, x_k, A] \, \leq \, 4(1 + \exp(B)) \, \frac{\log L}{\sqrt{L}}. 
\end{split}
\label{eq39}
\end{equation}
By a similar argument we can show that,
\begin{equation}
\nonumber\begin{split}
\mathbb{E}[\, (\operatorname{logit}( \overline{n}_{k}) - x_k)^2 \,| \, x_k, A] \, \leq \,  \resizebox{0.4\columnwidth}{!}{$\displaystyle 8(1 + \exp(B))^2 \, \left(\frac{\log L}{\sqrt{L}} \right)^2 $}.
\end{split}
\end{equation}

Thus, the expectation in Eq. (\ref{eq35}) is bounded as:
\vspace{-2mm}
\begin{align}
\nonumber \left| \mathbb{E} [ (\operatorname{logit} (\overline{n}_{k}))^2 \, - x_k^2 \, | \, A]  \right| &
\, \leq \, 8(1 + \exp(B)) \, \frac{\log L}{\sqrt{L}}  \\
& \hspace{-8mm} \times \resizebox{0.35\columnwidth}{!}{$\displaystyle \left( B + (1 + \exp(B)) \, \frac{\log L}{\sqrt{L}} \right) $} . 
\label{eq41}
\end{align}
Following a similar argument, one can show for $n \neq m$,
\begin{align}
\nonumber \big| \mathbb{E} [ \operatorname{logit} (\overline{n}_{k}) & \operatorname{logit} (\overline{n}_{m}) \, - x_k x_m \, | A\, ] \big| \, \leq \, 8(1 + \exp(B)) \\
& \times \resizebox{0.42\columnwidth}{!}{$\displaystyle \frac{\log L}{\sqrt{L}} \, \left( B + 2 (1 + \exp(B)) \, \frac{\log L}{\sqrt{L}} \right) $}. 
\label{eq42}
\end{align}
Finally, using the bounds of Eqs. (\ref{eq41}) and (\ref{eq42}) and noting that $\sum_{k = 0}^{K-1} (\nu_k^{(p)})^2 = 1$, for all $0 \leq p \leq P$, we can upper bound the expectation in Eq. (\ref{eq28}) as,
\begin{equation}
\begin{split}
|\mathbb{E}[\widehat{S}^{\sf mt}(\omega) \, - \, S^{\sf mt}(\omega)\, | \, A]|  \; \leq \;   g_1 K \, \frac{\log L}{\sqrt{L}} \,,
\end{split}
\label{eq43}
\end{equation}
where
\begin{equation}
\nonumber g_1 := \resizebox{0.9\columnwidth}{!}{$\displaystyle 8(1 + \exp(B)) \left\{ \left( \frac{1}{K} + 1 \right) B + \left( \frac{1}{K} + 2 \right) (1 + \exp(B)) \frac{\log L}{\sqrt{L}} \right\}.$}
\end{equation}
This concludes the proof of the bound on bias. Following similar bounding techniques, the second moment in Eq. (\ref{eq27}) can be bounded by:
\begin{equation}
\begin{split}
&\sqrt{\mathbb{E}[|\widehat{S}^{\sf mt}(\omega) \, - \, S^{\sf mt}(\omega)|^2 \, | \, A]} \; \leq \; g_2 K \, \frac{\log L}{\sqrt{L}},
\end{split}
\label{eq44}
\end{equation}
where
\begin{align}
\nonumber g_2 &:= 4 (1 + \exp(B)) \Bigg\{ \frac{\sqrt{2}}{K} \left[ \sqrt{\frac{13}{3}} \frac{\log L}{\sqrt{L}} \, (1 + \exp(B)) + B\right]  \\ 
\nonumber & \qquad \qquad \, + \left[ 4 \left(\frac{\log L}{\sqrt{L}} \, (1 + \exp(B)) + B\right)^2 - B^2 \right]^{1/2} \Bigg\}.
\end{align}
This concludes the proof of Theorem \ref{thm1}. 
\end{proof}

\section{Proof of Corollaries \ref{corollary1} and \ref{corollary2}}\label{Appendix_B}

\begin{proof}[Proof of Corollary \ref{corollary1}]
Proof of Corollary \ref{corollary1} follows the proof of Theorem \ref{thm1} closely, with the natural extension to the multivariate case. Following the proof of Theorem \ref{thm1}, the constants $g'_1$ and $g'_2$ in this case are given by:
\begin{equation}
\nonumber g'_1 := 8 \, (1 + \exp(B)) \times \left( B + 2(1 + \exp(B)) \frac{\log L}{\sqrt{L}} \right)
\end{equation}
\begin{equation}
\nonumber g'_2 := 4 (1 + \exp(B)) \sqrt{4 \left(\frac{\log L}{\sqrt{L}} \, (1 + \exp(B)) + B\right)^2 - B^2}
\end{equation}
\end{proof}

\begin{proof}[Proof of Corollary \ref{corollary2}]
As for Corollary \ref{corollary2}, we work under the technical assumptions of Theorems 1 and 2 in \cite{Das}. Following \cite{Das}, we assume that in Eq. (\ref{eq10}), $\mathbf{Q}_m = \mathbf{Q}$ for all $m$ in this proof, and that the EM algorithm finds estimates of $\mathbf{Q}$ and $\alpha$ close to their true value (for large enough $K$). Then, under Assumptions (1) and (2), we identify the effective observation $\widetilde{\mathbf{y}}_m^{(p)}$ corresponding to the $p^{\sf th}$ taper at window $m$ in \cite{Das} by the concatenation of $\nu_{w}^{(p)} \operatorname{logit}\big( \bar{n}_{(m-1)W + w,j}\big)$ for $w=1,2,\cdots,W$ and $j=1,2,\cdots,J$ in a vector of length $WJ$. We also assume, without loss of generality that $W = u N$, for some integer $u$. Then, we denote by $\bm{\Sigma}_\infty$ the steady state covariance of the backward smoother, and $\bm{\Lambda} := \alpha \bm{\Sigma}_\infty (\alpha^2 \bm{\Sigma}_\infty + \mathbf{Q})^{-1}$ and $\bm{\Gamma} = (\alpha^2 \bm{\Sigma}_\infty + \mathbf{Q})[\mathbf{I} - uW((\alpha^2 \bm{\Sigma}_\infty + \mathbf{Q})^{-1} + uW\mathbf{I})^{-1}]$, as in \cite{Das}. Note that these matrices are $NJ \times NJ$ in our case. Under the same assumptions \cite{Das}, we consider them to be diagonal with the $i^{\sf th}$ diagonal element being $\gamma_i$ and $\eta_i$ respectively. Then, following the proof of Theorem \ref{thm1} and those of Theorems 1 and 2 in \cite{Das}, it can be shown that the statement of the corollary holds with the constants:

\vspace{-5mm}
\begin{align}
\nonumber g''_1(\omega_n) & := \resizebox{0.6\columnwidth}{!}{$\displaystyle 8 (1 + \exp(B))\left( B + 2(1 + \exp(B)) \frac{\log L}{\sqrt{L}} \right)$}  \\
\nonumber &\times \eta_{(n-1)J + r} \, \eta_{(n-1)J + t}\sum_{s, s' = 1}^{M} \gamma_{(n-1)J + r}^{|s-m|} \gamma_{(n-1)J + t}^{|s'-m|},
\end{align}
\vspace{-2mm}
\begin{align}
\nonumber g''_2(\omega_n) &:= 4(1 + \exp(B)) \, \resizebox{0.5\columnwidth}{!}{$\displaystyle \sqrt{4 \left(\frac{\log L}{\sqrt{L}} \, (1 + \exp(B)) + B\right)^2 - B^2}$}\\
\nonumber & \quad \times \resizebox{0.7\columnwidth}{!}{$\displaystyle \eta_{(n-1)J + r} \, \eta_{(n-1)J + t}\sum_{s, s' = 1}^{M} \gamma_{(n-1)J + r}^{|s-m|} \gamma_{(n-1)J + t}^{|s'-m|} $},
\end{align}
\vspace{-4mm}

\noindent and
\vspace{-2mm}
\begin{align}
\nonumber \kappa_m(\omega_n) := \resizebox{0.8\columnwidth}{!}{$\displaystyle  \eta_{(n-1)J + r} \; \eta_{(n-1)J + t} \sum_{s, s' = 1}^{M} \gamma_{(n-1)J + r}^{|s-m|} \, \gamma_{(n-1)J + t}^{|s'-m|} \, \alpha^{|s-s'|}.$}
\end{align}
\vspace{-4mm}

\end{proof}

\vspace{-2mm}
\bibliographystyle{IEEEtran}
\bibliography{DynPPMTM}

\end{document}